\documentclass[11pt,onecolumn,draftcls]{IEEEtran}

\usepackage{amssymb, amsmath, amsthm, bbm, color}
\usepackage{mathrsfs}
\usepackage{subfig}
\usepackage{tikz}

\newcommand{\GF}{\mathrm{GF}}
\newcommand{\GL}{\mathrm{GL}}
\newcommand{\inn}{\mathrm{in}}
\newcommand{\ind}{\mathrm{ind}}
\newcommand{\outn}{\mathrm{out}}

\newcommand{\cp}{\mathrm{cp}}

\def\red{\textsc{red}}
\def\dim{\textsc{dim}}
\def\mindim{\textsc{mindim}}
\newcommand{\Fix}{\mathrm{Fix}}

\theoremstyle{plain}

\newtheorem{corollary}{Corollary}
\newtheorem{lemma}{Lemma}
\newtheorem{proposition}{Proposition}
\newtheorem{theorem}{Theorem}
\newtheorem*{claim*}{Claim}
\newtheorem{claim}{Claim}

\theoremstyle{definition}
\newtheorem{definition}{Definition}

\newtheorem{example}{Example}

\title{Reduction and Fixed Points of Boolean Networks and Linear Network Coding Solvability}
\author{Maximilien Gadouleau,~\IEEEmembership{Member, IEEE},\thanks{M. Gadouleau is with School of Engineering and Computing Sciences, Durham University, UK. \texttt{m.r.gadouleau@durham.ac.uk}}
 Adrien Richard,\thanks{A. Richard is with Laboratoire I3S, CNRS \& Universit\'e de Nice-Sophia Antipolis, France. \texttt{richard@unice.fr}}
and Eric Fanchon\thanks{E. Fanchon is with Universit\'e de Grenoble - CNRS, TIMC-IMAG UMR 5525, Grenoble, France. \texttt{eric.fanchon@imag.fr}}
\thanks{This work is partially supported by CNRS and The Royal Society through the International Exchanges Scheme grant {\em Boolean networks, network coding and memoryless computation.}}}

\begin{document}

\maketitle

\begin{abstract}
Linear network coding transmits data through networks by letting the intermediate nodes combine the messages they receive and forward the combinations towards their destinations. The solvability problem asks whether the demands of all the destinations can be simultaneously satisfied by using linear network coding. The guessing number approach converts this problem to determining the number of fixed points of coding functions $f:A^n\to A^n$ over a finite alphabet $A$ (usually referred to as Boolean networks if $A = \{0,1\}$)  with a given interaction graph, that describes which local functions depend on which variables. In this paper, we generalise the so-called reduction of coding functions in order to eliminate variables. We then determine the maximum number of fixed points of a fully reduced coding function, whose interaction graph has a loop on every vertex. Since the reduction preserves the number of fixed points, we then apply these ideas and results to obtain four main results on the linear network coding solvability problem. First, we prove that non-decreasing coding functions cannot solve any more instances than routing already does. Second, we show that triangle-free undirected graphs are linearly solvable if and only if they are solvable by routing. This is the first classification result for the linear network coding solvability problem. Third, we exhibit a new class of non-linearly solvable graphs. Fourth, we determine large classes of strictly linearly solvable graphs. 
\end{abstract}

\section{Introduction} \label{sec:intro}

\subsection{Background: network coding solvability and coding functions}

Network coding is a technique to transmit information through networks, which can significantly improve upon routing in theory \cite{ACLY00, YLCZ06}. At each intermediate node $v$, the received messages $x_{u_1}, \ldots, x_{u_k}$ are combined, and the combined message $f_v(x_{u_1},\ldots,x_{u_k})$ is then forwarded towards its destinations. The main problem is to determine which functions $f_v$ can transmit the most information. In particular, the \textbf{network coding solvability problem} tries to determine whether a certain network situation, with a given set of sources, destinations, and messages, is solvable, i.e. whether all messages can be transmitted to their destinations. This problem being very difficult, different techniques have been used to tackle it, including matroids \cite{DFZ06}, Shannon and non-Shannon inequalities for the entropy function \cite{DFZ07, Rii07}, error-correcting codes \cite{GR11}, and closure operators \cite{Gad13, Gad14}.  As shown in \cite{Rii07, Rii07a}, the solvability problem can be recast in terms of fixed points of (non-necessarily Boolean) networks.

\textbf{Boolean networks} have been used to represent a network of interacting entities as follows. A network of $n$ automata has a state $x= (x_1,\dots, x_n) \in \{0,1\}^n$, represented by a Boolean variable $x_i$ on each automaton $i$, which evolves according to a deterministic function $f = (f_1,\dots,f_n) : \{0,1\}^n \to \{0,1\}^n$, where $f_i : \{0,1\}^n \to \{0,1\}$ represents the update of the local state $x_i$. Boolean networks have been used to model gene networks \cite{Kau69, Tho73, TK01, Jon02}, neural networks \cite{MP43, Hop82, Gol85}, social interactions \cite{PS83, GT83} and more (see \cite{TD90, GM90}). Their natural generalisation where each variable $x_i$ can take more than two values in some finite alphabet $A$ has been investigated since this can be a more accurate representation of the phenomenon we are modelling \cite{TK01,KS08}. In order to avoid confusion, and despite the popularity of the term ``Boolean network,'' we shall refer to any function $f: A^n \to A^n$ as a \textbf{coding function}.

The structure of a coding function $f: A^n \to A^n$ can be represented via its \textbf{interaction graph} $G(f)$, which indicates which update functions depend on which variables. More formally, $G(f)$ has $\{1,\dots,n\}$ as vertex set and there is an arc from $j$ to $i$ if $f_i(x)$ depends essentially on $x_j$. In different contexts, the interaction graph is known--or at least well approximated--, while the actual update functions are not. One main problem of research on (non-necessarily Boolean) coding functions is then to predict their dynamics according to their interaction graphs. 

Among the many dynamical properties that can be studied, \textbf{fixed points} are crucial because they represent stable states; for instance, in the context of gene networks, they correspond to stable patterns of gene expression at the basis of particular biological processes. As such, they are arguably the property which has been the most thoroughly studied. The study of the number of fixed points and its maximisation in particular is the subject of a stream of work, e.g. in \cite{Rob86,ADG04a,RRT08,Ara08,Ric09,ARS14,GRR14}. In particular, a lot of literature is devoted to determining when a Boolean coding function admits multiple fixed points (see \cite{Ric13} for a survey).

The network coding solvability problem can be recast in terms of fixed points of coding functions as follows \cite{Rii07, Rii07a}. The so-called \textbf{guessing number} \cite{Rii07} of a digraph $G$ is the logarithm of the maximum number of fixed points over all coding functions $f$ whose interaction graph is a subgraph of $G$: $G(f) \subseteq G$. The guessing number is always upper bounded by the size of a minimum feedback vertex set of $G$; if equality holds, we say that $G$ is solvable and the coding function $f$ reaching this bound is called a solution. Then, a network coding instance $N$ is solvable if and only if some digraph $G_N$ (to be defined later) related to the instance $N$ is solvable.

\textbf{Linear network coding} is the most popular kind of network coding, where the intermediate nodes can only perform linear combinations of the packets they receive \cite{LYC03}. The network coding instance $N$ is then linearly solvable if and only if $G_N$ admits a linear solution. Many interesting classes of linearly solvable digraphs have been given in the literature (see \cite{Rii06, GR11}). However, as we shall explain in Section \ref{sec:NC_LS},  all the linearly solvable undirected graphs $G$ known so far are ``easily'' solved, because they are all vertex-full: the vertex set can be partitioned into $\alpha(G)$ cliques, where $\alpha(G)$ is the independence number of $G$ \cite{CM11}.

\subsection{Our approach and contribution}

Fixed points of coding functions and network coding are very closely linked; for instance \cite{GRR14} uses techniques from network coding and coding theory to derive bounds on the number of fixed points of specific coding functions. As such, in this paper we will derive \textbf{results of interest for both communities}. More precisely, we expand a new technique to study the number of fixed points of coding functions and we apply it to the solvability problem. Recently, \cite{NRTC11} introduced the reduction of coding functions in order to reduce the number of interacting automata while preserving some key dynamical properties. More precisely, for any loopless vertex $v$ of $G(f)$ the $v$-reduction of $f$ is obtained by evaluating $f_v$ and then replacing its expression instead of $x_v$ into all the other local functions $f_i$. The $v$-reduction notably preserves the number of fixed points \cite{NRTC11}. A very similar reduction procedure was proposed in the context of systems of differential equations \cite{KCA03}; this procedure is also based on variable elimination and preserves the number of fixed points.

In this paper, we generalise the concept of reduction of a coding function by a vertex in two fashions. We consider successive reductions vertex per vertex, and we prove in Theorem \ref{th:reduction_f} that this is equivalent to reducing all these vertices at once, provided that they induce an acyclic subgraph of the interaction graph. Since the reduction of a coding function has the same number of fixed points as the original coding function, we can then study the number of fixed points of fully reduced coding functions. We also introduce the concept of reduction of digraphs; again this can be done one vertex at a time or all at once, according to Theorem \ref{th:graph_reduction}. The interaction graph of a reduced coding function is then a subgraph of the reduction of its interaction graph. In particular, reducing an entire maximal acyclic set of a digraph yields a digraph with a loop on each vertex. Similarly, we can always successively reduce a coding function to one whose interaction graph has a loop on each vertex.  We then fully determine the maximum number of fixed points of coding functions for a given interaction graph with a loop on each vertex in Theorem~\ref{th:h_loops}.

We then apply this reduction approach to network coding solvability and derive four main results.
\begin{enumerate}
	\item We consider solvability by non-decreasing coding functions, which naturally extend routing. We show in Theorem \ref{th:monotone} that a digraph is solvable by a non-decreasing coding function if and only if it is solvable by routing. 
	
	\item We derive some important classification results for undirected graphs. We exhibit in Theorem \ref{th:solvable_not_vf} the first example of a non-vertex-full linearly solvable graph. We obtain in Theorem~\ref{th:SLS_weakly} a necessary condition for a graph $G$ to be strictly linearly solvable, i.e. to have a linear solution $f$ with $G(f) = G$. Using this condition, we then prove in Theorem \ref{th:triangle_free} that a triangle-free undirected graph is linearly solvable if and only if it is vertex-full; we also prove that all strictly linearly solvable complements of triangle-free graphs are vertex-full in Theorem \ref{th:alpha2}. For triangle-free graphs, our results indicate that the instance is linearly solvable if and only if it is solvable by routing; in other words, linear network coding does not help to solve these graphs. 
	
	\item Using Theorem \ref{th:SLS_weakly}, we exhibit in Theorem \ref{th:non_solvable_digraph} a new class of digraphs which are not linearly solvable. This is significant because few non-linearly solvable classes of digraphs are known so far, and proving non-linear solvability usually requires different techniques, such as graph entropy \cite{Rii06, CM11} or digraph closure \cite{Gad14}.
	
	\item We show in Theorem \ref{th:SLS_strongly} that a large class of digraphs are strictly linearly solvable. Strictly linearly solvable digraphs are not only interesting for some applications of coding functions (see Section \ref{sec:SLS}), but they also represent network coding instances where no arc is detrimental to the transmission of information.
\end{enumerate}

The rest of the paper is organised as follows. Section \ref{sec:network_reduction} studies the reduction of coding functions and the reduction of graphs and relates these two notions. Reductions of coding functions are then related to their fixed points in Section \ref{sec:fixed_points}. Finally, we apply the theory of coding function and graph reductions to the problem of linear network coding solvability in Section \ref{sec:NC_LS}.

\section{Reduction of coding functions} \label{sec:network_reduction}

\subsection{Definitions} \label{sec:network_reduction_definition}

We first review some concepts relating to coding functions. Let $V$ be a finite set, possibly empty, of cardinality $n$. Let $A$ be a finite set, referred to as the {\bf alphabet}, of cardinality $q \ge 2$; depending on the context, we will consider $A = \GF(q)$ or $A = \mathbb{Z}_q$ or $A = [q] := \{0,\dots,q-1\}$. Let $f : A^V \to A^V$
be a {\bf coding function} of dimension $\dim(f)=n$. We shall usually simplify notation and identify $A^V$ with $A^n$. We can then view $f: A^n \to A^n$ as $f = (f_1,\dots,f_n)$ where $f_v : A^n \to A$. For any $x \in A^n$ and any $I \subseteq V$, we also denote $x_{V \setminus I}$ as $x_{-I}$; we will usually identify a vertex $v$ with its corresponding singleton $\{v\}$.

A digraph with vertex set $V$ is a pair $G = (V,E)$ where $E \subseteq V^2$; we set $\dim(G)=|V| = n$. If $E$ is a symmetric set, we say that $G$ is undirected (i.e. we identify undirected and bidirected graphs). We associate with $f$ the digraph $G(f)$, referred to as the {\bf interaction graph} of $f$, defined by: the vertex
set is $V$; and for all $u,v\in V$, there exists an arc $(u,v)$ if and
only if $f_v$ depends essentially on $x_u$, i.e. there exist $x,y\in A^n$ that
only differ by $x_u\neq y_u$ such that $f_v(x)\neq f_v(y)$. We denote the set of all coding functions $f: A^n \to A^n$ for some $A$ of size $q$ with interaction graph $G$ as $F(G,q)$. 

We now review some basic concepts and introduce some notation for digraphs $G = (V,E)$ \cite{BG09a}. An {\bf induced subgraph} of $G$ is obtained by removing vertices of $G$; a {\bf spanning subgraph} of $G$ is obtained by removing arcs. If $I\subseteq V$, we denote by $G[I]$ the subgraph of $G$ induced by $I$, and we set $G\setminus I=G[V\setminus I]$. If $G[I]$ has no cycle, then we say that $I$ is an {\bf acyclic set}. An acyclic set $I = \{i_1,\dots, i_m\}$ can be sorted in {\bf topological order}, where $(i_k, i_l) \in G$ only if $k < l$. Thus if $I$ is an acyclic set of $G(f)$, then $f_{i_k}$ does not depend on the variables $i_l$ with $l>k$ and we can write $f_{i_k}(x)=f_{i_k}(x_{-I},x_{i_1},\dots,x_{i_{k-1}})$. The complement of an acyclic set is a {\bf feedback vertex set}. We denote the size of a minimum feedback vertex set of $G$ as $k(G)$ and the size of a maximum acyclic set of $G$ as $\alpha(G)$; we then have $\alpha(G) = n - k(G)$.

The {\bf in-neighbourhood} of a vertex $i$ in $G$ is denoted as $\inn_G(i) := \{u \in V : (u,i) \in G\}$; its {\bf in-degree} is $\ind_G(i) = |\inn_G(i)|$; when there is no ambiguity, we shall remove the dependence in $G$. The {\bf out-neighbourhood} and {\bf out-degree} are defined similarly. Paths and cycles are always supposed to be directed. If $s=(s_1,\dots, s_k)$ is a sequence of
distinct vertices of $G$, then $\{s\}=\{s_1,\dots,s_k\}$ denotes the {\bf support} of $s$. 

\begin{definition}[\cite{NRTC11}] \label{def:network_reduction}
For any $v \in V$ without a loop in $G(f)$, the {\bf $v$-reduction of $f$} is the coding function $f^{-v}:A^{V\setminus v}\to A^{V\setminus v}$, where for all $i \ne v$ and $x\in A^V$
$$
	f_i^{-v}(x_{-v}) := f_i(x_{-v}, f_v(x_{-v})).
$$
If $G(f)$ has a loop on $v$ then $f^{-v}=f$ by convention. 
\end{definition}

Thus $\dim(f^{-v})=\dim(f)-1$ if and only if $G(f)$ has no loop on $v$. Let $s=(s_1,s_2,\dots ,s_k)$ be a sequence of distinct vertices of $V$ of
length $|s|=k>0$. We write
\[
	f^{-s}=f^{-s_1s_2\dots s_k}=(f^{-s_1})^{-s_2})^{\dots})^{-s_{k}}.
\]
The sequence $s$ is a {\bf{reduction sequence of
\boldmath$f$\unboldmath}} if: $G(f)$ has no loop on $s_1$, and $G(f^{-s_1\dots
s_{r-1}})$ has no loop on $s_r$ 
for each $1<r\leq k$. So $s$ is a reduction sequence if and only if
\[
\dim(f^{-s})=\dim(f)-|s|.
\]
By convention the empty sequence $\epsilon$ is a reduction sequence, and $f^{-\epsilon}=f$.

\begin{definition} \label{def:cumulative}
Let $I = \{i_1,\dots,i_m\}$ be an acyclic set of $G(f)$ in topological order. We denote the {\bf cumulative $f$-coding function on $I$} as $F^I : A^{V\setminus I} \to A^{I}$ defined as
\begin{align*}
	F^I_{i_1} (x_{-I}) &:= f_{i_1}(x_{-I})\\
	F^I_{i_2} (x_{-I}) &:= f_{i_2}(x_{-I}, F^I_{i_1}(x_{-I}))\\
	&\vdots\\
	F^I_{i_m}(x_{-I}) &:= f_{i_m}( x_{-I}, F^I_{I \setminus i_m} (x_{-I}) ).
\end{align*}

The {\bf $I$-reduction of $f: A^{V} \to A^{V}$} is defined as the coding function $f^{-I}:A^{V\setminus I}\to A^{V\setminus I}$ such that
$$
	f_i^{-I}(x_{-I}) := f_i(x_{-I}, F^I(x_{-I})).
$$
\end{definition}

\begin{theorem} \label{th:reduction_f}
If $I$ is an acyclic set of $G(f)$, then any enumeration $s$ of $I$ is a reduction sequence of $f$ such that $f^{-s} = f^{-I}$.
\end{theorem}

\begin{proof}
We prove that if there is no arc from $v$ to $u$ and no loop on either vertex, then $f^{-uv} = f^{-vu}$. 
By direct application of the reduction rule, we have for all $i \notin \{u,v\}$,
\begin{align*}
	f^{-uv}_i(x_{-uv})
	&=f^{-u}_i(x_{-uv}, f^{-u}_v(x_{-uv}))\\
	&=f_i(x_{-uv}, f^{-u}_v(x_{-uv}),f_u(x_{-uv}, f^{-u}_v(x_{-uv})))\\
	&=f_i(x_{-uv}, f_v(x_{-uv},f_u(x_{-uv})),f_u(x_{-uv}, f^{-u}_v(x_{-uv})))
\end{align*}
and since there is no arc from $v$ to $u$ we get
\begin{align*}
	f^{-uv}_i(x_{-uv})
	&=f_i(x_{-uv}, f_v(x_{-uv},f_u(x_{-uv})),f_u(x_{-uv})).
\end{align*}
Again by direct application of the reduction rule, we have for all $i \notin \{u,v\}$,
\begin{align*}
	f^{-vu}_i(x_{-uv})
	&=f^{-v}_i(x_{-uv}, f^{-v}_u(x_{-uv}))\\
	&=f_i(x_{-uv}, f^{-v}_u(x_{-uv}),f_v(x_{-uv}, f^{-v}_u(x_{-uv})))
\end{align*}
and since there is no arc from $v$ to $u$ we have $f^{-v}_u(x_{-uv})=f_u(x_{-uv})$ thus
\begin{align*}
	f^{-vu}_i(x_{-uv})
	&=f_i(x_{-uv}, f_u(x_{-uv}),f_v(x_{-uv}, f_u(x_{-uv}))).
\end{align*}
Thus $f^{-uv}_i=f^{-vu}_i$ and the claim is proved.

Let $I = \{i_1,\dots,i_m\}$ in topological order and let $s$ and $t$ be enumerations of $I$. Firstly, suppose that $s$ and $t$ only differ by  a transposition of adjacent vertices, say $s = (s_1,\dots, s_m)$ and $t =(s_1,\dots,s_{k-2},\allowbreak s_k,s_{k-1},s_{k+1},\dots, s_m)$. We then have 
$$
	f^{-s_1 \dots s_k} = h^{-s_{k-1} s_k} = h^{-s_k s_{k-1}} = f^{-t_1 \dots t_k},
$$
where $h = f^{-s_1 \dots s_{k-2}}$, and hence $f^{-s} = f^{-t}$. Secondly, in the general case, it is well known that $t$ can be obtained from $s$ by transposing adjacent vertices: indeed the Coxeter generators of $I$ generate the symmetric group on $I$. Thus $f^{-s} = f^{-t}$; in particular, if $s$ is a topological order of $I$, then we obtain $f^{-I}$ described above.
\end{proof}

\begin{corollary}\label{cor:support_network}
If $s$ and $t$ are two reduction sequences of $f$ with the same acyclic support then $f^{-s} = f^{-t}$.
\end{corollary}

A coding function $h$ is a {\bf reduced form of
{\boldmath$f$\unboldmath}} if there exists a reduction sequence $s$
such that $f^{-s}=h$. A {\bf minimal reduced form of
{\boldmath$f$\unboldmath}} is a reduced form $h$ such that every
vertex of $G(h)$ has a loop. The set of reduced forms of $f$ is
denoted $\red(f)$. We are particularly interested in finding, according to $G(f)$, reduced forms of dimension as small as possible. In the ideal case, we
would like to obtain reduced forms of dimension 
\[
	\mindim(f) := \min_{h\in\red(f)}\dim(h).
\]

\subsection{Graph reduction} \label{sec:graph_reduction}

\begin{definition} \label{def:G-v}
If $G$ has no loop on $v$, we call
{\bf\boldmath$v$-reduction of $G$\unboldmath}, and we denote by
$G^{-v}$, the graph obtained from $G\setminus v$ by adding an arc $(u,w)$ (not already present) whenever $(u,v)$ and $(v,w)$ are arcs of $G$. By convention, if $G$ has a loop on $v$, then $G^{-v}=G$. 
\end{definition}

We shall use similar notation to that of the reduction of coding functions. A sequence $s=(s_1,\dots,s_k)$ of vertices of $G$ is a {\bf{reduction sequence of
\boldmath$G$\unboldmath}} if: $G$ has no loop on $s_1$, and $G^{-s_1\dots
s_{r-1}}$ has no loop on $s_r$ for each $1<r\leq k$. So $s$ is a reduction sequence if and only if $G^{-s}$ has $|V|-k$ vertices. 


\begin{definition} \label{def:G-I}
For any acyclic set $I$ of $G$, the {\bf $I$-reduction of $G$} is the digraph $G^{-I} := (V \setminus I, E')$, where $(u,w) \in E'$ if and only if either $(u,w) \in E$ or there is a path in $G$ from $u$ to $w$ through $I$ (that is, a path from $u$ to $w$ whose internal vertices are all in $I$).
\end{definition}

\begin{theorem} \label{th:graph_reduction}
If $I$ is an acyclic set of $G$, then any enumeration $s$ of $I$ is a reduction sequence of $G$ such that $G^{-s} = G^{-I}$.
\end{theorem}

\begin{proof}
The structure of the proof is similar to that of Theorem \ref{th:reduction_f}. We first prove that if $u,v \in V$ induce an acyclic subgraph, then $G^{-uv} = G^{-vu}$. Say that there is no arc from $v$ to $u$ and that there is no loop on either $u$ or $v$. Let us simplify notation and denote the proposition $(x,y) \in G$ as $xy$ and the proposition $(x,y) \in G^{-z}$ as $xy^{-z}$ for any vertices $x$, $y$, and $z$. Then for any $a,b \notin \{u,v\}$,
\begin{align*}
	ab^{-u} &\iff ab \vee \left( au \wedge ub \right),\\
	(a,b) \in G^{-uv} &\iff ab^{-u} \vee \left( av^{-u} \wedge vb^{-u} \right)\\
	&\iff ab \vee \left( au \wedge ub \right) \vee \left\{ \left[ av \vee \left( au \wedge uv \right) \right] \wedge  vb \right\}\\
	&\iff ab \vee \left( au \wedge ub \right) \vee \left( av \wedge vb \right) \vee \left( au \wedge uv \wedge vb \right).
\end{align*}
Similarly,
\begin{align*}
	ab^{-v} &\iff ab \vee \left( av \wedge vb \right),\\
	(a,b) \in G^{-vu} &\iff ab^{-v} \vee \left( au^{-v} \wedge ub^{-v} \right)\\
	&\iff ab \vee \left( av \wedge vb \right) \vee \left\{ au \wedge  \left[ub \vee \left( uv \wedge vb \right) \right] \right\}\\
	&\iff ab \vee \left( av \wedge vb \right) \vee \left( au \wedge ub \right) \vee \left( au \wedge uv \wedge vb \right).
\end{align*}

Now let $I = \{i_1,\dots,i_m\}$ in topological order and let $s$ and $t$ be enumerations of $I$. Firstly, suppose that $s$ and $t$ only differ by  a transposition of adjacent vertices, say $s = (s_1,\dots, s_m)$ and $t =(s_1,\dots,s_{k-2},\allowbreak s_k,s_{k-1},s_{k+1},\dots, s_m)$. We then have 
$$
	G^{-s_1 \dots s_k} = H^{-s_{k-1} s_k} = H^{-s_k s_{k-1}} = G^{-t_1 \dots t_k},
$$
where $H = G^{-s_1 \dots s_{k-2}}$, and hence $G^{-s} = G^{-t}$. Secondly, in the general case, it is well known that $t$ can be obtained from $s$ by transposing adjacent vertices: indeed the Coxeter generators of $I$ generate the symmetric group on $I$; thus $G^{-s} = G^{-t}$. 

In particular, we prove that if $s$ is the topological order, then we obtain $G^{-I}$ described above. The proof is by induction on $|I|$. For $|I| = 1$, the result is obvious. Suppose the result holds for all induced acyclic subgraphs of size $m-1$. Let $s = i_1,\dots,i_m$ be a topological order of $I$. By definition, we have $(u,v) \in G^{-s}$ if and only if $(u,v) \in G^{-i_1,\dots,i_{m-1}}$ or $(u,i_m), (i_m,v) \in G^{-i_1,\dots,i_{m-1}}$. Thus, by induction hypothesis, $(u,v) \in G^{-s}$ if and only if $(u,v) \in G^{-(I\setminus i_m)}$ or $(u,i_m), (i_m,v) \in G^{-(I\setminus i_m)}$. This is equivalent to
\begin{itemize}
	\item either, $(u,v) \in G$;
	
	\item or $G$ has a path from $u$ to $v$ through $I \setminus i_m$;
	
	\item or $(u,i_m), (i_m,v) \in G$;
	
	\item or $G$ has a path from $u$ to $i_m$ through $I \setminus i_m$ and $(i_m,v) \in G$;
	
	\item or $(u,i_m)\in G$ and $G$ has a path from $i_m$ to $v$ through $I \setminus i_m$ (impossible, since $\outn(i_m) \cap I = \emptyset$);
	
	\item or $G$ has a path from $u$ to $i_m$ through $I\setminus i_m$ and a path from $i_m$ to $v$ through $I \setminus i_m$ (also impossible).
	
\end{itemize}
This is clearly equivalent to either $(u,v) \in G$ or there is a path in $G$ from $u$ to $v$ through $I$. Thus $(u,v)\in G^{-s}$ if and only if $(u,v)\in G^{-I}$.
\end{proof}

We make three remarks on the reduction of digraphs.
\begin{enumerate}
	\item If $s$ and $t$ are two reduction sequences of $G$ with the same acyclic support then $G^{-s} = G^{-t}$.
	
	\item $G^{-I}$ has a loop on each vertex if and only if $I$ is a maximal acyclic set. Therefore, there is a bijection between the set of minimal reduced forms of $G$ and the set of its minimal feedback vertex sets. Since it is well known that
finding a minimum feedback vertex set is an NP-Complete problem, finding a
minimum reduced form is also NP-Complete.
	
	\item For any $G$ and any acyclic set $I$, we have $G \setminus I \subseteq G^{-I}$.  We prove below a converse to this result: depending on the initial digraph $G$, the reduced form $G^{-I}$ of $G$ may add any possible arc to $G \setminus I$.
\end{enumerate}



\begin{proposition} \label{prop:D_H}
For any digraph $D$ with vertex set $J$ and any spanning subgraph $H$ of $D$, there exists a set $I$ and a digraph $G$ with vertex set $I \cup  J$  such that $G \setminus I = G[J] = H$ and $G^{-I} = D$.
\end{proposition}

\begin{proof}
Let $I$ be the set of arcs in $D$ but not in $H$ and let $G$ be the graph on $I \cup J$ such that $G[J] = H$ and for any arc $e = (u,v) \in I$, $G$ contains the arcs $(u,e)$ and $(e,v)$. Then it is clear that $G^{-I} = D$. 
\end{proof}

\subsection{Interaction graph of the reduced coding function} \label{sec:G_reduction}

The reduction of digraphs yields an estimate on the interaction graph of the reduction of coding functions.

\begin{proposition}\label{pro:subgraph1}
If $I$ is an acyclic set of $G(f)$ then $G(f^{-I})$ is a subgraph of $G(f)^{-I}$.
\end{proposition}

\begin{proof}
We only prove that if $G(f)$ has no loop on $v$, then $G(f^{-v})$ is a subgraph of
$G(f)^{-v}$; the result is an easy
consequence. We have
$$
	f_w^{-v}(x) = f_w(x_{-v}, f_v(x_{-v})),
$$
hence $(u,w)$ is an arc in $G(f^{-v})$ only if either $(u,w)$ is already in $G(f)$ or $(u,v), (v,w) \in G(f)$, i.e. $(u,w) \in G(f)^{-v}$.
\end{proof}

\begin{corollary}\label{pro:subgraph2}
Every reduction sequence of $G(f)$ is a reduction sequence of $f$.
\end{corollary}

According to Proposition~\ref{pro:subgraph1}, we have $\mindim(f)\leq k(G(f))$. We show that this bound on $\mindim(f)$ is the best possible as a
function of $G(f)$. The {\bf min-net} of a digraph $G$ over $[q] = \{0,\dots,q-1\}$ ($q \ge 2$) is the coding function $f := \min(G,q)$ defined as
$$
	f_i(x) :=	\min\{ x_j : j \in \inn_G(i) \}
$$ 
with the convention that $\min(\emptyset)=q-1$.

\begin{proposition} \label{prop:reduction_min}
For any digraph $G$ and any acyclic set $I$ of $G$, $\min(G,q)^{-I} = \min(G^{-I}, q)$. Therefore, $\mindim(\min(G,q)) = k(G)$.
\end{proposition}

\begin{proof}
Again, we only prove the case where $I$ is only one vertex $v$. If $f = \min(G,q)$, we have for all $i \ne v$
$$
	f_i^{-v}(x) = \min\{x_j : j \in \inn_G(i) \setminus v, \min\{x_k : k \in \inn_G(v)\} \} = \min \{ x_k : k \in \inn_{G^{-v}}(i) \}.
$$
Thus $\min(G,q)^{-v} = \min(G^{-v},q)$.
\end{proof}

Note that according to the preceding, finding a minimum reduced form
of a min-net is equivalent to finding a minimum vertex set in a
digraph. So finding a minimum form of a coding function is NP-Hard.

Although there exists a coding function (the min-net) whose reductions follow the reductions of its interaction graph, we prove in the following two propositions that in general we cannot say much about the interaction graph of the reduced coding function. First, we derive the analogue of Proposition \ref{prop:D_H} for the interaction graphs of coding functions.

\begin{proposition} \label{prop:any_graph_reduction}
Let $D$ and $H$ be any digraphs with vertex set $J$. Then for any $q \ge 2$ there exists a set $I$ and a coding function $f: A^{I \cup J} \to A^{I \cup J}$ such that $G(f)[J] = D$ and $G(f)^{-I} = H$.
\end{proposition}

\begin{proof}
Let $I = I_D \cup I_H$, where $I_D$  is the set of all arcs in $D$ but not in $H$ and $I_H$ is the set of all arcs in $H$ but not in $D$. Then let $G$ be the graph with vertex set $I \cup J$ such that $G[J] = D$ and for any $(u,v) \in I$, $(u,(u,v)), ((u,v),v) \in G$. For any $(u,v) \in I_D$ and any $x \in [q]^n$ (with $n = |I \cup J|$), we denote 
$$
	y_{(u,v)} = x_u + q - 1 - x_{(u,v)}
$$
and for any $j \in J$, $y^j$ as the state with coordinates $y_{(u,j)}$ for all $u \in \inn_D(j) \setminus \inn_H(j)$.

Finally, let $f: [q]^n \to [q]^n$  be defined as
$$
	f_j(x_{\inn_D(j) \setminus \inn_H(j)}, x_{I_D}, x_{\inn_D(j) \cap \inn_H(j)}, x_{I_H}) = \min (y^j, x_{\inn_D(j) \cap \inn_H(j)}, x_{I_H \cap \inn_G(j)})
$$
for all $j \in J$ and
$$
	f_{(u,v)}(x) = x_u
$$
for all $(u,v) \in I$. It is clear that $f \in F(G,q)$, hence $G(f)[J] = D$. Moreover, reducing $I_D$ yields
\begin{align*}
	f_j^{-I_D}(x_{\inn_D(j)}, x_{I_H}) &= \min (q-1,\dots,q-1, x_{\inn_D(j) \cap \inn_H(j)}, x_{I_H \cap \inn_G(j)}) = \min (x_{\inn_D(j) \cap \inn_H(j)}, x_{I_H \cap \inn_G(j)}),
\end{align*}
and then reducing $I_H$ yields
$$
	f_j^{-I}(x_J) = \min (x_{\inn_D(j) \cap \inn_H(j)}, x_{\inn_H(j) \setminus \inn_D(j) }) = \min(x_{\inn_H(j)}),
$$
thus $G(f^{-I}) = H$.
\end{proof}

Second, we prove that even reducing a single vertex may in fact remove any set of arcs from the original interaction graph.

\begin{proposition} \label{prop:vanish_graph_reduction}
Let $G$ be a digraph with a vertex $v$ such that $\inn(v) = \outn(v) = V \setminus v$, and let $G$ have minimum in-degree at least 2. Then for any $q \ge 2$ and any spanning subgraph $H$ of $G \setminus v$, there exists a coding function $f \in F(G,q)$ such that $G(f) = G$  and $G(f^{-v}) = H$.
\end{proposition}

\begin{proof}
Say $v = n$ and for all $i$, let $y_i = \min\{x_i,1\} \in \{0,1\}$. We define the function for the vertex $n$:
$$
	f_n(x) = \bigvee_{i =1}^{n-1} y_i.
$$
That way, we can focus on each vertex of $H$ separately; without loss we only consider the vertex $1$. Let $N = \inn_G(1) \setminus n$, $P = \inn_H(1)$ and $Q = N \setminus P$; then
\begin{align*}
	f_1(x) &= \left( \bigwedge_{p \in P} y_p \right) \land \left( y_n \lor \bigvee_{q \in Q} \neg y_q \right),\\
	f_1^{-n}(x) &= \left( \bigwedge_{p \in P} y_p \right) \land \left( \bigvee_{i =1}^{n-1} y_i \lor \bigvee_{q \in Q} \neg y_q \right)
	= \bigwedge_{p \in P} y_p,
\end{align*}
with the convention that an empty conjunction is equal to $1$ and an empty disjunction is equal to $0$.
\end{proof}

%
%
%

We finish this section with an example illustrating the reduction of graphs and coding functions.

\begin{example} \label{ex:reduction}
Consider the following coding function $f: \{0,1\}^4 \to \{0,1\}^4$ given by
\begin{align*}
	f_1(x) &=x_3\land (x_2\lor x_4)\\
	f_2(x) &=x_1 \lor x_4\\
	f_3(x) &=x_2\\
	f_4(x) &=x_3.
\end{align*}
Then $f^{-4}$ is given by:
\begin{align*}
	f^{-4}_1(x_{-4}) &=x_3\land (x_2\lor x_3)=x_3\\
	f^{-4}_2(x_{-4}) &=x_1 \lor x_3\\
	f^{-4}_3(x_{-4}) &=x_2.
\end{align*}
Thus $G(f^{-4})$ is a strict subgraph of $G(f)^{-4}$, as seen on Figure \ref{fig:G(f)}. However, $f^{-34} = f^{-43}$ is given by
\begin{align*}
	f^{-34}_1(x_1,x_2) &= x_2\\
	f^{-34}_2(x_1,x_2) &= x_1 \lor x_2,
\end{align*}
and hence $G(f^{-34}) = G(f)^{-34}$. Finally, $f^{-134} = f^{-431}$ is given by
$$
	f^{-134}_2(x_2) = x_2 \lor x_2 = x_2
$$
and $G(f^{-134})$ only has one vertex with a loop.

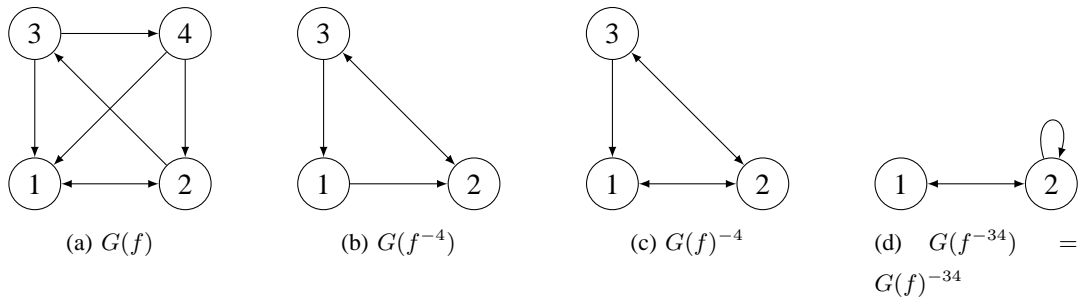
\begin{figure}[Ht]
\centering
	\subfloat[$G(f)$]
{\begin{tikzpicture}
 \tikzstyle{every node}=[draw,shape=circle];

	\node (1) at (0,0) {1};
	\node (2) at (2,0) {2};
	\node (3) at (0,2) {3};
	\node (4) at (2,2) {4};

	\draw[latex-latex] (2) -- (1);
	\draw[-latex] (3) -- (1);
	\draw[-latex] (4) -- (1);
	
	\draw[-latex] (4) -- (2);
	
	\draw[-latex] (2) -- (3);
	
	\draw[-latex] (3) -- (4);
\end{tikzpicture}} \hspace{1cm} \subfloat[$G(f^{-4})$]
{\begin{tikzpicture}
 \tikzstyle{every node}=[draw,shape=circle];

	\node (1) at (0,0) {1};
	\node (2) at (2,0) {2};
	\node (3) at (0,2) {3};

	\draw[-latex] (3) -- (1);
	
	\draw[-latex] (1) -- (2);
	\draw[latex-latex] (3) -- (2);
\end{tikzpicture}} \hspace{1cm} \subfloat[$G(f)^{-4}$]
{\begin{tikzpicture}
 \tikzstyle{every node}=[draw,shape=circle];

	\node (1) at (0,0) {1};
	\node (2) at (2,0) {2};
	\node (3) at (0,2) {3};

	\draw[latex-latex] (2) -- (1);
	\draw[-latex] (3) -- (1);
	
	\draw[latex-latex] (3) -- (2);
\end{tikzpicture}} \hspace{1cm} \subfloat[$G(f^{-34}) = G(f)^{-34}$]
{\begin{tikzpicture}
 \tikzstyle{every node}=[draw,shape=circle];

	\node (1) at (0,0) {1};
	\node (2) at (2,0) {2};

	\draw[latex-latex] (2) -- (1);

	\draw[-latex] (2) .. controls (1.7,1) and (2.3,1) .. (2);
\end{tikzpicture}}
\caption{Example of coding function and graph reduction.} \label{fig:G(f)}
\end{figure}

\end{example}

\section{Fixed points of coding functions} \label{sec:fixed_points}

\subsection{Maximum number of fixed points} \label{sec:max_fixed_points}

The $q$-{\bf guessing number} \cite{Rii07} and $q$-{\bf strict guessing number} of $G$ are respectively defined as
\begin{align*}
	g(G,q) &:= \log_q \max\{|\Fix(f)| : f \in F(G',q), G' \subseteq G\},\\
	h(G,q) &:= \log_q \max\{|\Fix(f)| : f \in F(G,q)\}.\\
\end{align*}
The guessing number of loopless digraphs was thoroughly investigated in \cite{Rii07, WCR09, GR11, CM11, BCDRV13, CDR14}; the strict guessing number is new. We first relate $h(G,q)$ to $g(G,q)$.

\begin{lemma} \label{lem:g>1}
If $g(G,q) \ge 1$, then $h(G,q) \ge 1$.
\end{lemma}

\begin{proof}
If $g(G,q) \ge 1$, then $G$ has a cycle. Say that the vertices $0$ up to $l-1$ form a chordless cycle and consider the coding function $f \in F(G,q)$ defined by
$$
	f_i(x) = \begin{cases}
		x_{i-1 \mod l} &\text{if } x_j\geq x_{i-1 \mod l}\text{ for all }j\in\inn(i)\\
		x_{i-1 \mod l} + 1 &\text{otherwise}
	\end{cases}
$$
if $0 \le i \le l-1$ and
$$
	f_j(x) = \min\{x_k : k \in \inn(j)\}
$$
otherwise. Then it is clear that if $G$ is of minimal in-degree at least one then for any $a \in [q]$, $(a, \dots, a)$ is fixed by $f$. Thus $h(G,q) \ge \log_q |\Fix(f)| \ge 1$. Otherwise, let $I_0$ be the set of vertices of in-degree $0$, and for $0<k<n$, let $I_k$ be the set of vertices $i$ such that $\inn(i)\subseteq I_0\cup\dots\cup I_{k-1}$. Then, for any $a\in [q]$, the point  $x^a\in[q]^n$ such that $x^a_i=q-1$ if $i\in I_k$ for some $k$ and $x^a_i=a$ otherwise is a fixed point of $f$, thus $h(G,q) \ge \log_q |\Fix(f)| \ge 1$.
\end{proof}

%

\begin{proposition} \label{prop:g<h<g}
For all $q \ge 2$ and any digraph $G$,
$$
	g(G,q) \ge h(G,q) \ge  g(G,q-1) \log_q(q-1) \ge g(G,q) - n \log_q \left( 1+ \frac{1}{q-1} \right).
$$
\end{proposition}

\begin{proof}
The first inequality is trivial. We now prove the second. Let $f: [q-1]^n \to [q-1]^n$ with interaction graph $G' \subseteq G$ and with $(q-1)^{g(G,q-1)}$ fixed points. Let $f'\in F(G,q)$ such that 
$$
	f'_v(x) = \begin{cases} 
	f_v(x) & \mbox{if } x_{\inn(v)} \in [q-1]^{\ind(v)}\\
	q & \mbox{otherwise}.
	\end{cases}
$$
Then $\Fix(f) \subseteq \Fix(f')$ and hence $(q-1)^{g(G,q-1)} = |\Fix(f)| \le |\Fix(f')| \le q^{h(G,q)}$.

Let us now prove the third inequality. Let $f: [q]^n \to [q]^n$ with interaction graph $G' \subseteq G$ and with $q^{g(G,q)}$ fixed points. Then for any vertex $v$ and any permutation $\pi$ of $[q]$, consider the coding function $f^{v,\pi}$ defined as 
$$
	f^{v,\pi}_u(x) = \begin{cases}
	\pi( f_v( \pi^{-1}(x_v), x_{-v})) & \mbox{if } u=v\\
	f_u( \pi^{-1}(x_v), x_{-v}) & \mbox{otherwise}.
	\end{cases}
$$
Then $x \in \Fix(f)$ if and only if $(\pi(x_v),x_{-v}) \in \Fix(f^{v,\pi})$ and hence $|\Fix(f^{v,\pi})| = q^{g(G,q)}$. Denote
\begin{align*}
	R(v,a) 	&:= |\{ x \in \Fix(f) : x_v = a \}| = |\{ x \in \Fix(f^{v,\pi}) : x_v = \pi(a) \}|,\\
	r(v)	&:= \min_{a \in [q]} R(v,a) \le q^{-1} q^{G(G,q)}.
\end{align*}
Consider a permutation $\sigma$ of $[q]$ such that $r(v) = R(v,\sigma^{-1}(q-1))$;  we then obtain 
\begin{align*}
	|\{ x \in \Fix(f^{v,\sigma}) : x_v \in [q-1]\}| &= |\{ x \in \Fix(f) : x_v \in \sigma^{-1}([q-1])\}| \\
	&=  \sum_{a \ne \sigma^{-1}(q-1)} R(v,a) \\
	&=  |\Fix(f)| - R(v,\sigma^{-1}(q-1))\\
	&\ge (1 - q^{-1}) q^{g(G,q)}.
\end{align*}
Thus, $f^{v,\sigma}$ has at least $(1 - q^{-1}) q^{G(G,q)}$ fixed points with $x_v \in [q-1]$. Applying this strategy recursively for all $n$ vertices, we find that there exists a coding function with at least $(1-q^{-1})^n q^{g(G,q)}$ fixed points in $[q-1]^n$. By considering the restriction of this coding function to $[q-1]^n$, we obtain $(q-1)^{g(G,q-1)} \ge (1-q^{-1})^n q^{g(G,q)}$.
\end{proof}

\begin{corollary}
We have $\lim_{q \to \infty} h(G,q) = \lim_{q \to \infty} g(G,q) = H(G)$ for all $G$, where $H(G)$ is the entropy of $G$ \cite{Rii06}.
\end{corollary}

\subsection{Fixed points and reduction} \label{sec:fixed_points_reduction}

\begin{proposition}[See \cite{NRTC11}] \label{prop:fixed_points_reduction}
Let $f$ be a coding function and $h$ be a reduced form of $f$. With the
convention that $f$ has a unique fixed point if $\dim(f)=0$, $f$ and $h$ have the same number of fixed points.
\end{proposition}

\begin{proof}
Again, we can assume that $h = f^{-v}$ for some vertex $v$ without a loop in $G(f)$. We then have
$f_i(x) = x_i$ for all $i$ if and only if $f_v(x) = x_v$ and $f_i^{-v}(x) = f_i(x_{-v}, f_v(x)) = f_i(x) = x_i$ for all $i \ne v$.
\end{proof}

\begin{example}
Let $f$ be the coding function in Example \ref{ex:reduction}. The fixed points of $f$ and its successive reductions are respectively given by
\begin{align*}
	\Fix(f) &= \{(0,0,0,0), (1,1,1,1)\},\\
	\Fix(f^{-4}) &= \{(0,0,0), (1,1,1)\},\\
	\Fix(f^{-34}) &= \{(0,0), (1,1)\},\\
	\Fix(f^{-134}) &= \{0,1\},
\end{align*}
and successive reductions preserve the number of fixed points. In particular, since $k(G(f)) = 1$ and $f^{-134}$ is the identity of $A^{k(G(f))}$, Proposition \ref{prop:fixed_points_reduction} indicates that $f$ is indeed a solution for $G(f)$.
\end{example}

Let $S = V \setminus I$ be a feedback vertex set of $G(f)$. Then
according to Proposition~\ref{pro:subgraph1}, $f$
has a reduced form $f^{-I}$ with dimension $|S|$. So $f^{-I}$ has obviously at
most $q^{|S|}$ fixed points, and since $f$ and $f^{-I}$ have the same number
of fixed points, $f$ has at most $q^{|S|}$ fixed points. This provides
an alternative proof of (a modified form of) a theorem of
Aracena \cite{Ara08} (see Riis \cite{Rii06}): If $S$ is a feedback vertex set of $G(f)$, then $f$
has at most $q^{|S|}$ fixed points. In other words, $h(G,q) \le k(G)$; by obvious extension, we obtain $g(G,q) \le k(G)$ as well.

In particular, if $G(f)$ has no cycle, then $f$ has a reduced
form $f^{-V}$ of dimension zero. Then $f^{-V}$ has a unique fixed point and
we deduce that $f$ has a unique fixed point. This provides an
alternative proof of a theorem of Robert \cite{Rob95}: If $G(f)$ is acyclic, then $f$ has a unique fixed point.

As seen below, we cannot say anything interesting about the guessing number of reduced digraphs in general.

\begin{proposition} \label{prop:guessing_reduction}
The guessing number of $G$ and that of its reduction $G^{-v}$ are related as follows.
\begin{enumerate}
	\item Let $G$ be a digraph and $v$ a vertex of $G$, then $g(G,q) \le g(G^{-v},q)$ for all $q \ge 2$. 
	
	\item If $G$ is acyclic, then $h(G,q) = g(G,q) = g(G^{-v},q) = h(G^{-v},q) = 0$. 
	
	\item For any $n \ge 3$ there exists $G$ on $n$ vertices such that $g(G,q) = h(G,q) = 1$, $h(G^{-v},q) = \log_q(q^{n-1}-1)$ and $g(G^{-v},q) = n-1$ for all $q$.
\end{enumerate}
\end{proposition}

\begin{proof}
The first two statements are clear. Now consider the complete bipartite graph $G = K_{n-1,1}$, also called the star on $n$ vertices, where $n$ is the centre of the star. Then this vertex forms a feedback vertex set and hence $g(G,q) = 1$, and by Lemma \ref{lem:g>1}, $h(G,q) = 1$. Then $G^{-n}$ is the complete graph on $n-1$ vertices with a loop on each vertex, and we have $g(G^{-n},q)$ and $h(G^{-n},q)$ from Theorem \ref{th:h_loops} and Example \ref{ex:h_loop} below.
\end{proof}

\subsection{Fixed points of fully reduced coding functions} \label{sec:fixed_points_fully_reduced}

We are then interested in studying the number of fixed points of coding functions which are fully reduced, i.e. whose interaction graphs have a loop on each vertex. For any loopless digraph $G$, we denote the graph obtained from $G$ by adding a loop on each vertex as $\mathring{G}$. Clearly, $g(\mathring{G},q) = n$; moreover, $h(\mathring{G},q) = n$ if and only if $G$ is empty (this is the interaction graph of the identity function).

For any loopless $G$, an {\bf in-dominating set} (IDS) is a set of vertices $X \subseteq V$ such that for all $v \in V$ with positive in-degree, either $v \in X$ or $\inn(v) \cap X \ne \emptyset$. Denote the number of IDSs of $G$ of size $k$ as $I_k(G)$; clearly, $I_n(G) = 1$.

\begin{theorem} \label{th:h_loops}
For any loopless graph $G$, 
$$
	h(\mathring{G},q) = \log_q \sum_{k=0}^n (q-1)^k I_k(G).
$$
\end{theorem}

\begin{proof}
For any property $\mathcal{P}$, we denote the function which returns $1$ if $\mathcal{P}$ is satisfied and $0$ otherwise as ${\mathbbm 1}\{\mathcal{P}\}$. Also, we write $\inn(i)$ and $\outn(i)$ for $\inn_{\mathring{G}}(i)$ and $\outn_{\mathring{G}}(i)$ so that $i\in\inn(i)\cap \outn(i)$. We define the coding function $g \in F(\mathring{G},q)$ as
$$
	g_i(x) := \begin{cases}
	x_i & \mbox{if } \ind(i) = 1\\
	x_i + {\mathbbm 1}\{ x_{\inn(i)} = (0,\ldots,0) \}\mod q & \mbox{otherwise}.
	\end{cases}
$$
For any $x$, $x = g(x)$ if and only if $\{v \in V : x_v \ne 0\}$ is an in-dominating set. This proves the lower bound.

Now let $f$ with $G(f) = \mathring{G}$ and $q^{h(\mathring{G},q)}$ fixed points. Any local function of $f$ is expressed as
$$
	f_i(x) = \begin{cases}
	a(x_i) & \mbox{if } \ind(i) = 1\\
	x_i + e_i(x_{\inn(i)}) \mod q& \mbox{otherwise},
	\end{cases}
$$
where $e_i(x_{\inn(i)}) = f_i(x) - x_i \mod q$. It is clear that the optimal choice for the function $a$ is simply $a(x_i) = x_i$. Therefore, we only focus on the case where $\ind(i) \ge 2$ henceforth.

We now show that we can always assume that $e_i$ takes a non-zero only once. Let $Y = \{y \in A^{\ind(i)} : e_i(y) \ne 0\}$ and let $y^i \in Y$. Now, let $f'$ such that $f'_j = f_j$ for all $j \ne i$ and
$$
	f'_i(x) = x_i + {\mathbbm 1}\{x_{\inn(i)} = y^i\}  \mod q.
$$
Suppose $f(x) = x$, then $f'_j(x) = f_j(x) = x_j$ for all $j \ne i$; moreover, $x_{\inn(i)} \notin Y$ hence $x_{\inn(i)} \ne y^i$ and $f'_i(x) = f_i(x) = x_i$. Therefore, $|\Fix(f')| \ge |\Fix(f)|$.

Hence we can consider $f'$ instead. We now show that choosing $y^i = (0,\ldots,0)$ for any vertex $i$ maximises the number of fixed points. Consider a vertex $k$ and define a new function $f''$ as
\begin{align*}
	f''_j &= f'_j \quad \forall\, j \notin \outn(k),\\
	f''_i(x) &= x_i + {\mathbbm 1}\{x_{\inn(i)} = z^i\} \mod q \quad \forall\, i \in \outn(k),
\end{align*}
where $z^i_j = y^i_j$ if $j \ne k$ and $z^i_k = 0$. 
Let $x'\in\Fix(f')\setminus \Fix(f'')$. Then there exists $i\in\outn(k)$ such that $z^i=x'_{\inn(i)}\neq y^i$. Defining $x''$ by only changing the $k$-coordinate of $x'$ to $x''_k := y^i_k$ we obtain $z^i\neq x''_{\inn(i)}= y^i$ and $z^j\neq x''_{\inn(j)}$ for all $j\in\outn(k)\setminus i$ (because $x''_k>0=z^j_k$). Thus $x''\in\Fix(f'')\setminus \Fix(f')$. Hence, there is an injection from $\Fix(f')\setminus \Fix(f'')$ to $\Fix(f'')\setminus \Fix(f')$, thus $f''$ has at least as many fixed points as $f'$.

Thus, we can always choose $z^i_k = 0$ for all $i$ and all $k$, which yields the coding function $g$. 
\end{proof}


\begin{corollary} \label{cor:h_loop}
For any loopless $G$, we have 
$$
	h(\mathring{G},q) \ge n \log_q (q-1) + \log_q \left(1 + \frac{n}{q-1} \right),
$$ 
and hence $\lim_{q \to \infty} h(\mathring{G},q) = n$.
\end{corollary}

\begin{proof}
For any $v \in V$, $V \setminus v$ is an IDS. Therefore, $I_{n-1}(G) = n$ and since $V$ is also an IDS, $I_n(G) = 1$. Therefore, $h(\mathring{G},q) \ge \log_q \left( n(q-1)^{n-1} + (q-1)^n \right)$.
\end{proof}

\begin{example} \label{ex:h_loop}
In general, computing the sum $\sum_k (q-1)^k I_k(G)$ is \#P-Complete. However, we can exhibit five special cases for which the formula is easy to derive; all graphs have vertex set $V = \{1,\dots,n\}$.
\begin{itemize}
	\item For the clique $K_n$ (with arcs $(i,j)$ for all $i \ne j$), 
	$$
		h(\mathring{K}_n, q) = \log_q( q^n - 1 ).
	$$
	
	\item For the transitive tournament $T_n$ (with arcs $(i,j)$ for all $i < j$),
	$$
		h(\mathring{T}_1,q) = 1 \quad \text{and} \quad h(\mathring{T}_n,q) = n-2 + \log_q(q^2-1) \,\, \forall \, n \ge 2.
	$$
	
	\item For the inward directed star $iS_n$ (with arcs $(i,n)$ for all $1 \le i \le n-1$), 
	$$
		h(\mathring{iS}_n, q) = \log_q( q^n -1 ).
	$$
	
	\item For the outward directed star $oS_n$ (with arcs $(n,i)$ for all $1 \le i \le n-1$), 
	$$
		h(\mathring{oS}_n, q) = \log_q( q^n - q^{n-1} + (q-1)^{n-1} ).
	$$
	
	\item For the undirected star (with arcs $(i,n)$ and $(n,i)$ for all $1 \le i \le n-1$), 
	$$
		h(\mathring{S}_n, q) = \log_q( q^n - q^{n-1} + (q-1)^{n-1} ).
	$$
\end{itemize}
\end{example}

\begin{proof}
For $K_n$, we have $I_0(K_n) = 0$ and $I_k(K_n) = \binom{n}{k}$ for all $1 \le k \le n$. Therefore, $\sum_k (q-1)^k I_k(K_n) = q^n - 1$. For $T_n$ ($n \ge 2$), a set of vertices $X$ is an in-dominating set if and only if it contains either the first or the second vertex. Therefore, $I_k(T_n) = \binom{n}{k} - \binom{n-2}{k}$ and $\sum_k (q-1)^k I_k(T_n) = q^n - q^{n-2}$. The proof for the stars is similar: we have $I_0(iS_n) = 0$ and $I_k(iS_n) = \binom{n}{k}$ for all $1 \le  k \le n$; we also have $I_{n-1}(S_n) = I_{n-1}(oS_n) = n$ and $I_k(S_n) = I_k(oS_n) = \binom{n-1}{k-1}$ otherwise.
\end{proof}

\section{Application to linear network coding solvability} \label{sec:NC_LS}

\subsection{Network coding solvability and guessing number} \label{sec:guessing_number}


We now apply the theory of coding function reduction to linear network coding solvability. The network coding solvability problem asks whether a given network coding instance is solvable, i.e. whether all messages can be transmitted to their destinations simultaneously. In particular, if the local functions $f_v$ are linear, then the instance is linearly solvable. For the study of solvability, any network coding instance can be converted into a \textbf{multiple unicast} without any loss of generality \cite{DZ06, Rii07}. A multiple unicast instance consists of  an acyclic network $N$ and a finite alphabet $A$ of cardinality $q$, where
\begin{itemize}
	\item each arc in the network carries an element of $A$;
	
	\item the instance is given in its so-called circuit representation, i.e. the same message flows on every arc coming out of the same vertex;
	
	\item the network has $k$ sources $s_1,\dots,s_k$, $k$ destinations $d_1,\dots,d_k$, and $\alpha$ intermediate nodes $i_{k+1}, \dots, i_{k+\alpha}$;
	
	\item each destination $d_i$ ($1 \le i \le k$) requests an element from $A$ from a corresponding source $s_i$.
\end{itemize}
This network coding instance is {\bf solvable} over $A$ if all the demands of the destinations can be satisfied at the same time.

The solvability of a multiple unicast instance can be decided by determining the guessing number of a related digraph. By merging each source with its corresponding destination node into one vertex, we form the digraph $G_N$ on $n := k + \alpha$ vertices. In general, we have $g(G_N,q) \leq k$ for all $q$ and the original network coding instance is solvable over $A$ if and only if $k(G_N) = k$ (this condition is purely graph-theoretic and does not involve coding functions; as such, we assume it is always satisfied) and $g(G_N,q) = k$, in which case we say that $G_N$ is solvable over $A$ \cite{Rii07} (an analogous result holds for linear solvability). Therefore, while network coding considers how the information flows from sources to destinations, the guessing number captures the intuitive notion of how much information circulates through the digraph. 

We illustrate the conversion of a network coding instance to a guessing number problem for the famous butterfly network in Figure \ref{fig:butterfly} below. It is well-known that the butterfly network is solvable over all alphabets, and conversely the clique $K_3$ has guessing number $2$ over any alphabet. The solutions are shown in Figure \ref{fig:butterfly} and indeed the operations done in the butterfly network correspond to the fixed point equations on the clique.

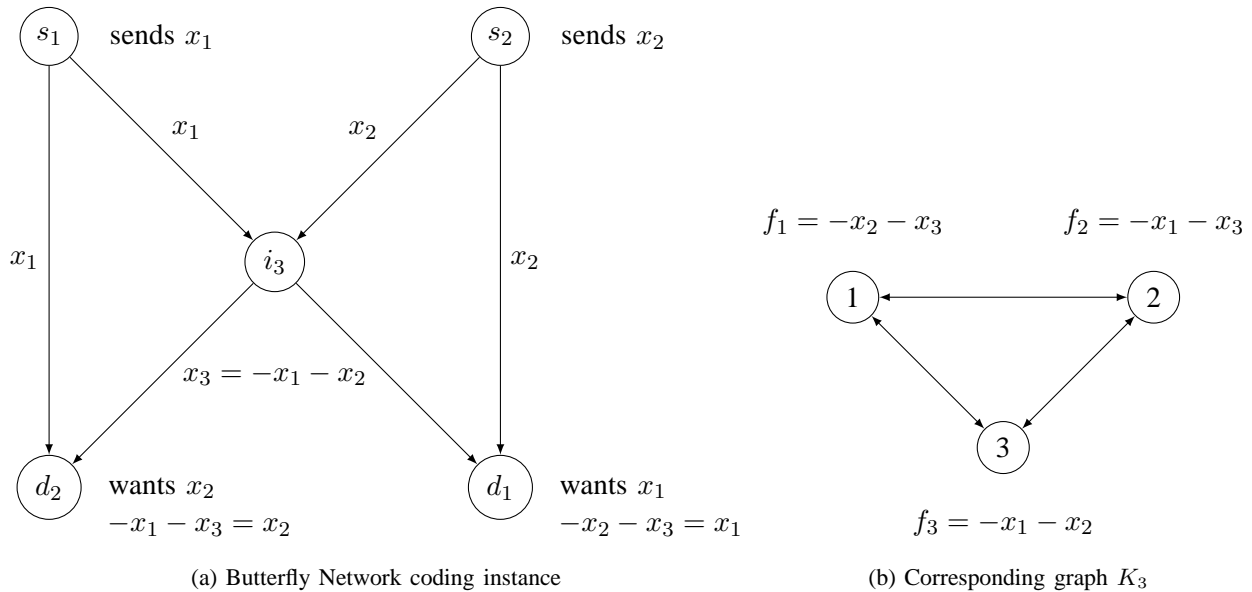
\begin{figure}
\centering
	\subfloat[Butterfly Network coding instance]
	{\begin{tikzpicture}[auto]
	
	\node[draw,shape=circle] (s1) at (0,6) {$s_1$};
	\node (source1) at (1.5,6) {sends $x_1$};
	
	\node[draw,shape=circle] (d1) at (6,0) {$d_1$};
	\node (destination1) at (7.5,0) {wants $x_1$};
	\node (dd1) at (8,-0.5) {$-x_2 - x_3 = x_1$};
	
	\node[draw,shape=circle] (s2) at (6,6) {$s_2$};
	\node (source2) at (7.5,6) {sends $x_2$};
	
	\node[draw,shape=circle] (d2) at (0,0) {$d_2$};
	\node (destination2) at (1.5,0) {wants $x_2$};
	\node (dd1) at (2,-0.5) {$-x_1 - x_3 = x_2$};
	
	\node[draw,shape=circle] (i3) at (3,3) {$i_3$};
	\node (ii3) at (3,1.5) {$x_3 = -x_1 - x_2$};
	
	\draw[-latex] (s1) to node [swap] {$x_1$} (d2);
	\draw[-latex] (s1) to node {$x_1$} (i3);
	\draw[-latex] (s2) to node {$x_2$} (d1);
	\draw[-latex] (s2) to node [swap] {$x_2$} (i3);
	
	\draw[-latex] (i3) -- (d1);
	\draw[-latex] (i3) -- (d2);
	\end{tikzpicture}} \subfloat[Corresponding graph $K_3$]
	{\begin{tikzpicture}	
	\node[draw,shape=circle] (x) at (0,7) {1};
	\node (xx) at (0,8) {$f_1 = -x_2 - x_3$};
	
	\node[draw,shape=circle] (y) at (4,7) {2};
	\node (yy) at (4,8) {$f_2 = -x_1 - x_3$};
	
	\node[draw,shape=circle] (z) at (2,5) {3};
	\node (zz) at (2,4) {$f_3 = -x_1 - x_2$};
	
	\draw[latex-latex] (x) -- (y);
	\draw[latex-latex] (x) -- (z);
	\draw[latex-latex] (y) -- (z);
	\end{tikzpicture}
	}
\caption{The butterfly network.}
\label{fig:butterfly}
\end{figure}

\subsection{Solvability by non-decreasing coding functions} \label{sec:NDS}

We first apply the reduction approach to network coding solvability by non-decreasing coding functions. Here, we consider $A = [q] = \{0, \dots, q-1\}$ with the usual linear order. We then say that a local function $f_v$ is non-decreasing if it is non-decreasing in every variable $x_u$; the coding function is {\bf non-decreasing} if all its local functions are non-decreasing. For instance, the min-net introduced in Section \ref{sec:G_reduction} is a non-decreasing coding function. Non-decreasing coding functions have been widely studied (see \cite{Ara08, Ric13, GRR14}); they are usually represented by an interaction graph with positive signs on all arcs (see \cite{Ric13} and the references therein for a survey of the work on signed interaction graphs). 

More closely related to network coding, {\bf routing} can be viewed as a non-decreasing coding function. Indeed, routing corresponds to local functions of the form $f_v(x_u)$ for some $u \in \inn(v)$. Routing then achieves a guessing number of $c(G)$, where $c$ is the maximum number of {\bf disjoint cycles} in $G$; thus a graph is solvable by routing if and only if $c(G) = k(G)$. It is shown in \cite{GRR14} that general non-decreasing coding functions can significantly outperform routing in terms of guessing number: for instance, on the clique $K_n$, routing achieves a guessing number of $c(K_n) = \lfloor n/2 \rfloor$, while non-decreasing functions achieve $n-3 - \epsilon$ when the alphabet is large enough \cite[Proposition 6]{GRR14}. However, Theorem \ref{th:monotone} below proves that non-decreasing functions do not outperform routing in terms of solvability.


\begin{theorem} \label{th:monotone}
For any digraph $G$, the following are equivalent:
\begin{enumerate}
	\item $G$ is solvable by a non-decreasing coding function over some alphabet.
	
	\item $G$ is solvable by a non-decreasing coding function over any alphabet.
	
	\item $G$ is solvable by routing.
	
	\item $c(G) = k(G)$.
\end{enumerate}
\end{theorem}

\begin{proof}
We only have to prove that the first property implies the fourth one. Let $f: [q]^n \to [q]^n$ be a non-decreasing coding function with $G(f) \subseteq G$ and $q^{k(G)}$ fixed points. Let $I$ be a maximal acyclic set of $G$ such that $|V\setminus I|=k(G)$. Since $f^{-I}$ and $f$ have the same number of fixed points, $f^{-I}$ is the identity on $[q]^{V\setminus I}$. 

For every vertex $v \notin I$, let $e_v \in [q]^{V\setminus I}$ be the $v$-th unit vector defined by $(e_v)_v = 1$ and $(e_v)_u = 0$ for all $u \ne v$; we set $\overline{e_v} = 1 - e_v$. Consider
$$
	C_v := \{v \} \cup \{i \in I: F^I_i(e_v) > F^I_i(\overline{e_v})\}. 
$$

\begin{claim*} \label{claim:disjoint}
For all distinct $u,v \notin I$, $C_u \cap C_v = \emptyset$.
\end{claim*}

\begin{proof}
For any $i \in I$, $F^I_i$ is a non-decreasing function of $x_{-I}$. Since $\overline{e_u} \ge e_v$, we obtain
for any $i \in C_u \cap C_v$:
$$
	F^I_i(e_u) > F^I_i(\overline{e_u}) \ge F^I_i(e_v) > F^I_i(\overline{e_v});
$$
yet $\overline{e_v} \ge e_u$ implies $F^I_i(\overline{e_v}) \ge F^I_i(e_u)$, which is the desired contradiction.
\end{proof}

\begin{claim*} \label{claim:cycle}
For all $v \notin I$, $C_v$ contains a cycle.
\end{claim*}

\begin{proof}
We only need to show that for all $i \in C_v$, $C_v \cap \inn(i) \ne \emptyset$. Firstly, let $i \in C_v \setminus \{v\}$, and suppose that $C_v \cap \inn(i) = \emptyset$, then $F_{\inn(i) \cap I}^I(\overline{e_v}) \ge F_{\inn(i) \cap I}^I(e_v)$ and $(\overline{e_v})_{-v} \ge (e_v)_{-v}$. Hence
$$
	F_i^I(\overline{e_v}) =  f_i \left( (\overline{e_v})_{-v}, F_{\inn(i) \cap I}^I(\overline{e_v}) \right)
	\ge f_i \left( (e_v)_{-v}, F_{\inn(i) \cap I}^I(e_v) \right)
	= F_i^I(e_v),
$$
which contradicts the fact that $i \in C_v$. Secondly, let $i = v$, and again suppose that $C_v \cap \inn(v) = \emptyset$; thus $F^I_{\inn(v) \cap I}(\overline{e_v}) \ge F^I_{\inn(v) \cap I}(e_v)$. Recall that $f^{-I}$ is the identity, hence
$$
	0 = f_v^{-I}(\overline{e_v}) = f_v^{-I} \left( (\overline{e_v})_{-v}, F^I_{\inn(v) \cap I}(\overline{e_v}) \right) \ge f_v^{-I} \left( (e_v)_{-v}, F^I_{\inn(v) \cap I}(e_v) \right) = f_v^{-I}(e_v) = 1.
$$
\end{proof}

By the claims above, we have $n - |I| = k(G)$ disjoint cycles in the graph $G$.
%
\end{proof}

\subsection{Linearly solvable undirected graphs} \label{sec:undirected}

We are now interested in linear coding functions. A {\bf linear coding function} is any coding function $f: R^V \to R^V$, where $R$ is a commutative ring of order $q$ and such that $f_i(x) = \sum_{u \in \inn(i)} a_{i,u} x_u$ for some $a_{i,u}$ invertible in $R$. For any $G$ we denote the set of linear coding functions with interaction graph $G$ over a commutative ring of order $q$ as $L(G,q)$. The set of fixed points of a linear coding function forms a submodule of $R^V$, hence we denote the $q$-{\bf linear guessing number} \cite{Rii07, GR11, CFHL14} and $q$-{\bf strict linear guessing number} of $G$ respectively as 
\begin{align*}
	g_L(G,q) &= \max \{ \mathrm{dim}\, \Fix(f) : f \in L(G',q), G' \subseteq G\},\\
	h_L(G,q) &= \max \{ \mathrm{dim}\, \Fix(f) : f \in L(G,q)\}.
\end{align*}
We say that a digraph $G$ is {\bf linearly solvable} if $g_L(G,q) = g(G,q) = k(G)$ for some $q$. We say it is {\bf strictly linearly solvable} if $h_L(G,q) = h(G,q) = k(G)$. It is easy to prove that $G$ is linearly solvable if and only if $G$ has a strictly linearly solvable spanning subgraph $H$ with $k(G) = k(H)$.

The minimum number of parts in any partition of the vertex set of $G$ into cliques is denoted as $\cp(G)$; if $G$ is undirected, then $\cp(G) = \chi(\bar{G})$, the chromatic number of its complement. We say that a digraph $G$ is {\bf vertex-full} ({\bf edge-full}, respectively) if all its vertices (arcs, respectively) can be covered by $\alpha(G)$ cliques. In other words, $G$ is vertex-full if and only if $\cp(G) = \alpha(G)$. Clearly, if $G$ is edge-full, then it is undirected; a characterisation of edge-full graphs is given in \cite{BD83} and we shall give another one in Proposition \ref{prop:edge-full} below.  

We can easily obtain a classical lower bound on the guessing number. Firstly, the clique $K_n$ is always linearly solvable over all alphabets by the following coding function $f$ (see Figure \ref{fig:butterfly}): 
$$
	f_i(x_{-i}) = - \sum_{j \ne i} x_j \mod q.
$$
Indeed, all states summing to zero mod $q$ are fixed by $f$ and hence $g_L(K_n,q) = g(K_n,q) = n-1$. (For $n=1$, we simply set $f(x) = 0$.) Therefore, if we partition the vertex set of $G$ into $\cp(G)$ cliques and apply the corresponding coding function on each clique, we obtain a linear coding function with $q^{n-\cp(G)}$ fixed points, thus yielding  \cite{CM11}
$$
	g_L(G,q) \ge n-\cp(G).
$$

This lower bound implies that vertex-full graphs are linearly solvable over all alphabets. On the other hand, many classes of linearly solvable digraphs are not vertex-full, e.g. the directed cycle (see \cite{GR11} for more striking examples). Until now, however, the only known linearly solvable undirected graphs are vertex-full. Based on the results in \cite{CDR14}, we can construct the first example of a linearly solvable undirected graph which is not vertex-full. Firstly, for two digraphs $G_1$ and $G_2$ on disjoint vertex sets of sizes $n_1$ and $n_2$ respectively, their {\bf bidirectional union} is $G := G_1 \bar{\cup} G_2$ where $G_1$ and $G_2$ are linked by all possible edges between them. The linear guessing number then satisfies for all $q$ \cite{GR11}
$$
	g_L(G,q) = \min\{n_1 + g_L(G_2,q), n_2 + g_L(G_1,q)\}.
$$

\begin{theorem} \label{th:solvable_not_vf}
There exists an undirected graph which is linearly solvable, and yet it is not vertex-full. 
\end{theorem}

\begin{proof}
Let $G_1 := E_6$ denote the empty graph on $n_1 := 6$ vertices and with linear guessing number $0$ for any $q$. Let $G_2 := \mathfrak{C}$ denote the Clebsch graph: $\mathfrak{C}$ has $n_2 := 16$ vertices, independence number $\alpha(\mathfrak{C}) = 5$, and $g_L(\mathfrak{C},3) \ge 10$ \cite{CDR14}. Since $\mathfrak{C}$ is triangle-free but not vertex-full, it is not linearly solvable as we shall see below. Nonetheless, the graph $G := E_6 \bar{\cup} \mathfrak{C}$ is linearly solvable but not vertex-full, since $n = 22$, $\alpha(G) = 6$ and $g_L(G,3) = 16$.
\end{proof}

\begin{definition} \label{def:compatible}
Let $I$ be a non-empty acyclic set $I$ of a digraph $G$.
\begin{itemize}
	\item $I$ is {\bf strongly compatible} if for all $u,v \notin I$, $(u,v) \in G$ if and only if there is a path from $u$ to $v$ through $I$. 
	\item $I$ is {\bf weakly compatible} if for all $u,v \notin I$, the following holds: if $(u,v)$ is an arc, then there is a path from $u$ to $v$ through $I$; otherwise, there is either no path from $u$ to $v$ through $I$ or there are at least two paths from $u$ to $v$ through $I$.
\end{itemize}
\end{definition}

\begin{theorem} \label{th:SLS_weakly}
If $G$ is strictly linearly solvable over some alphabet, then all maximum acyclic sets of $G$ are weakly compatible.
\end{theorem}

The proof of the theorem is based on the following lemma: if we consider the interaction graph $G(f)$ of a linear coding function $f$, we can only \emph{erase} an arc from $G(f)$ if we use a path through $I$.

\begin{lemma} \label{lem:contraction_linear}
Let $f$ be a linear coding function and $I$ be an acyclic set of $G(f)$, and any vertices $u,v$ outside of $I$ such that $(u,v) \in G(f)$ but $(u,v) \notin G(f^{-I})$. Then there exists a path in $G(f)$ from $u$ to $v$ through $I$.
\end{lemma}

\begin{proof}
Suppose $(u,v) \in G(f)$ but $(u,v) \notin G(f^{-I})$ and that there is no path in $G(f)$ from $u$ to $v$ through $I$. Denote $f_v(x) = \sum_{j \in \inn(v)} a_j x_j$. Denote $N = \inn(v) \cap I$, then there is no path in $G(f)$ from $u$ to $N$ through $I$; as such there is no arc from $u$ to $N$ in $G(f^{-(I \setminus N)})$. Thus, we have 
$$
	f_v^{-I}(x) = a_u x_u + \sum_{j \ne u} b_j x_j;
$$
the only occurrence of the variable $x_u$ is due to the original $f_v(x)$.
\end{proof}

\begin{proof}[Proof of Theorem \ref{th:SLS_weakly}]
Suppose that $I$ is a maximum acyclic set of $G$ which is not weakly compatible. There are two ways weak compatibility can be violated.
\begin{enumerate}
	\item Let $u,v \notin I$ such that $(u,v) \in G$ and yet there is no path from $u$ to $v$ through $I$. Then by Lemma \ref{lem:contraction_linear} for any linear coding function $f \in L(G,q)$, $(u,v)$ is an arc in $G(f^{-I})$, thus by Theorem \ref{th:h_loops} $f$ has fewer than $q^{k(G)}$ fixed points.
	
	\item Let $u,v \notin I$ such that $(u,v) \notin G$ and yet there is a unique path from $u$ to $v$ through $I$. If $f_a(x) = \sum_b c_{a,b} x_b$ for all $a$ and $b \in \inn(a)$ and if the path is $u_0 = u, u_1, \dots, u_k = v$, it is easy to check that the $x_u$ term in $f^{-I}_v$ is $\prod_{i=1}^k c_{u_i, u_{i-1}} \ne 0$. Again $f^{-I}$ is not the identity and $f$ has fewer than $q^{k(G)}$ fixed points.
\end{enumerate}
\end{proof}

Not all undirected graphs $G$ where all the maximum independent sets are weakly compatible are vertex-full. For instance, the bidirectional union $G = E_3 \bar{\cup} \bar{\mathfrak{G}}$ of an independent set of size three $E_3$ with the complement of the Gr\"otzsch graph $\bar{\mathfrak{G}}$ is a counter-example. The Gr\"otzsch graph is illustrated in Figure \ref{fig:non-VF}; it is triangle-free and has chromatic number $4$. Therefore, its complement is not vertex-full: $\alpha(\bar{\mathfrak{G}}) = 2$ while $\cp(\bar{\mathfrak{G}}) = 4$. In $G$, $E_3$ then forms a maximum independent set, which is clearly weakly compatible; however $\alpha(G) = 3$ while $\cp(G) = 4$, thus $G$ is not vertex-full.

Nonetheless, we can classify linearly solvable triangle-free undirected graphs. A {\bf matching} in a digraph is a union of disjoint undirected edges in the digraph. We denote the number of edges in a maximum matching in the digraph $G$ as $\mu(G)$. If $G$ is undirected, then it is easily seen that $c(G) = \mu(G)$; hence $G$ is solvable by routing if and only if $\mu(G) = k(G)$. Moreover, if $G$ is triangle-free, then these properties are in turn equivalent to $G$ being vertex-full. Theorem \ref{th:triangle_free} then proves that if an undirected triangle-free graph is solvable by linear network coding, then it is solvable by routing.

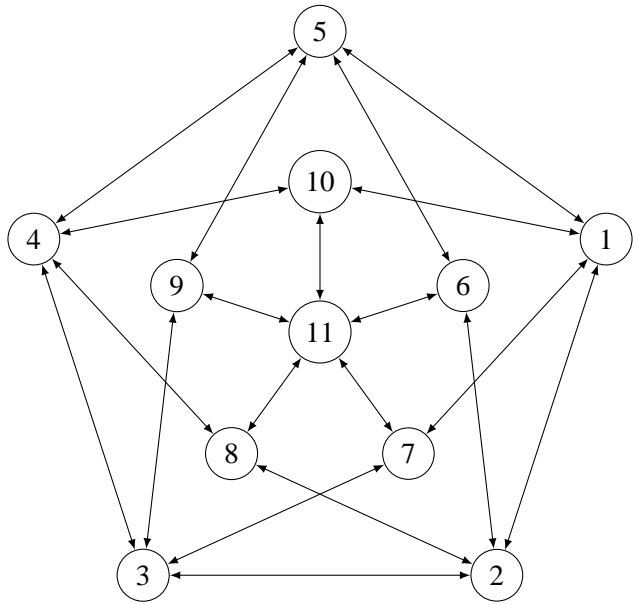
\begin{figure}[Ht]
\begin{center}
\begin{tikzpicture}
 \tikzstyle{every node}=[draw,shape=circle];

	\node (11)  at (0,0) {11};

\foreach \x in {6,...,10}{
	\node (\x) at (90-5*72-72*\x:2) {\x};
	\draw[latex-latex] (\x) -- (11);
}

\foreach \x in {1,...,5}{
	\node (\x) at (90-72*\x:4) {\x};
}

    \draw[latex-latex] (1) -- (2);
    \draw[latex-latex] (2) -- (3);
    \draw[latex-latex] (3) -- (4);
    \draw[latex-latex] (4) -- (5);
    \draw[latex-latex] (5) -- (1);

	\draw[latex-latex] (1) -- (7);
	\draw[latex-latex] (1) -- (10);

	\draw[latex-latex] (2) -- (8);
	\draw[latex-latex] (2) -- (6);

	\draw[latex-latex] (3) -- (9);
	\draw[latex-latex] (3) -- (7);

	\draw[latex-latex] (4) -- (10);
	\draw[latex-latex] (4) -- (8);

	\draw[latex-latex] (5) -- (6);
	\draw[latex-latex] (5) -- (9);
\end{tikzpicture}
\end{center}
\caption{The Gr\"otzsch graph $\mathfrak{G}$.} \label{fig:non-VF}
\end{figure}

\begin{theorem} \label{th:triangle_free}
Let $G$ be an undirected triangle-free graph. The following are equivalent:
\begin{enumerate}
	\item $G$ is linearly solvable over some alphabet.
	
	\item $G$ is linearly solvable over all alphabets.
	
	\item $G$ is solvable by routing.
	
	\item $G$ is vertex-full.
	
	\item $\mu(G) = k(G)$.
\end{enumerate}
\end{theorem}

\begin{proof}
We first remark that a triangle-free graph $G$ is vertex-full if and only if it has a matching of size $\mu(G) = k(G)$, in which case $G$ is linearly solvable (by routing) over all alphabets. Now, suppose $G$ is triangle-free and linearly solvable over some alphabet. Then there exists a subgraph $H$ of $G$ such that $H$ is strictly linearly solvable and $k(H) = k(G)$. $H$ is not necessarily undirected, hence we denote the undirected graph obtained from $H$ by adding any arc $(u,v)$ if $(v,u) \in H$ as $\bar{H}$; thus $\bar{H}$ is a spanning subgraph of $G$ with $k(\bar{H}) = k(G)$. Let $I$ be a maximum independent set of $G$, then $I$ is also a maximum acyclic set of $H$; by Theorem \ref{th:SLS_weakly}, it is weakly compatible in $H$. Then if $(u,v)$ is an arc in $H$ outside of $I$, there exists $i \in I$ such that $u$, $i$, and $v$ form a triangle in $\bar{H}$, which is impossible. Thus $\bar{H}$ is bipartite and by the K\"{o}nig-Egerv\'ary theorem \cite{BM08}, $\mu(\bar{H}) = k(\bar{H}) = k(G)$. Since $\mu(G) \ge \mu(\bar{H})$, we obtain that $\mu(G) = k(G)$. 
\end{proof}

%

We now prove that strictly linearly solvable complements of triangle-free graphs are vertex-full as well.

\begin{theorem} \label{th:alpha2}
Let $G$ be an undirected graph with $\alpha(G) = 2$. If $G$ is strictly linearly solvable over some alphabet, then $G$ is vertex-full (or equivalently, $G$ is the complement of a bipartite graph).
\end{theorem}

\begin{proof}
If $G$ is strictly linear solvable over some alphabet, then by Theorem \ref{th:SLS_weakly}, every non-edge is weakly compatible, and we prove that this implies that $G$ is vertex-full, i.e. that the vertex set of $G$ can be partitioned into two cliques. Let $ab$ be a non-edge in $G$, let $C_a$ be a maximal clique containing $a$, and let $C_b$ be a maximal clique containing $b$. If $C_a$ and $C_b$ cover all vertices, we are done. Otherwise, there exists $c$ which does not belong to either clique.

\begin{claim} \label{claim:1}
There exist $d \in C_a$, $e \in C_b$, disjoint from $a$ and $b$, such that $a,b,c,d,e$ induce a graph with exactly 7 edges and the following 3 non-edges: $ab$, $cd$ and $ce$.
\end{claim}

\begin{proof}
Since $a,b,c$ cannot form an independent set, without loss $ac$ is an edge. By maximality of $C_a$, there exists $d \in C_a$ such that $ad$ is an edge and $cd$ is a non-edge. Then $ab$ is weakly compatible, $cd$ is a non-edge, and they have a common neighbour (namely, $a$) in $ab$: $cd$ must have another common neighbour, namely $b$, which means that $bc$ and $bd$ are edges. In turn, there exists $e \in C_b$ such that $be$ is an edge and $ce$ is not an edge; as above, $ae$ is also an edge. Finally, since $c,d,e$ cannot form an independent set, $de$ is an edge. \end{proof}

\begin{claim} \label{claim:2}
If $c$ and $f$ do not belong to $C_a$ or $C_b$, then $cf$ is an edge.
\end{claim}

\begin{proof}
The vertices corresponding to $c$ are $a$ to $e$ as in Claim \ref{claim:1}; let $f$ not in $C_a$ or $C_b$ either and suppose that $cf$ is not an edge. Since $c,d,f$ cannot form an independent set, $fd$ is an edge. Thus there exists $g\in C_a$ with $q\neq d$ such that $fg$ is not an edge, and similarly there exists $h\in C_b$ with $h\neq e$ such that $fh$ is not an edge. Now, $cd$ is weakly compatible, $fg$ is not an edge and they only have $d$ as common neighbour in $cd$, which is a contradiction. 
\end{proof}

Therefore, the vertices outside of $C_a$ or $C_b$ form a clique, which we shall refer to as $C_c$.

\begin{claim} \label{claim:3}
Let $f \in C_c$ and $g \in C_a$ such that $fg$ is not an edge. Then for any $i \in C_b$, $gi$ is an edge.
\end{claim}

\begin{proof}
Suppose that $gi$ is not an edge. Since $f,g,i$ cannot form an independent set, $fi$ is an edge. Then using the notation above, $gi$ is weakly compatible, $fh$ is not an edge and they only have $i$ as a common neighbour in $gi$, which is a contradiction.
\end{proof}

By the above, the following two sets of vertices induce disjoint cliques and cover all vertices:
\begin{enumerate}
	\item $C_c$ and all the vertices in $C_a$ connected to all the vertices in $C_c$;
	
	\item $C_b$ and all the remaining vertices of $C_a$.
\end{enumerate}
%
\end{proof}

\subsection{Non linearly solvable digraphs} \label{sec:non-LS}

Theorem \ref{th:SLS_weakly} yields and easy way to construct digraphs that are not strictly linearly solvable. Indeed, let $I = \{i_1,\dots,i_m\}$ in topological order be a maximum acyclic set of $G$ and let $(u,v)$ be an arc outside of $I$ such that the out-neighbourhood of $u$ in $I$ is after (in topological order) than the in-neighbourhood of $v$. Then there is no path from $u$ to $v$ through $I$ and $I$ is not weakly compatible.

More interestingly, based on Theorem \ref{th:SLS_weakly}, we can construct digraphs which are not linearly solvable. The strategy to construct such a digraph $G$ uses two main ideas. Firstly, we force any possible linear solution to use an arc $(u,v)$ in a minimum feedback vertex set $J$. This can be done by ensuring that the graph obtained by removing $(u,v)$ has a smaller feedback vertex set than $G$. Secondly, we make sure that there is no path from $u$ to $v$ through the corresponding maximum acyclic induced subgraph $I = V \setminus J$. Thus, $I$ is not weakly compatible and by Theorem \ref{th:SLS_weakly}, the graph is not linearly solvable.

Let $G_k = (I \cup J, E)$ be any digraph such that
\begin{itemize}
	\item $I = \{i_1,\dots, i_{k-1}\}$ and $J = \{j_1,\dots, j_k\}$ are disjoint;
	
	\item $I$ is acyclic;
	
	\item $J \setminus \{j_k\}$ is acyclic; 
	
	\item $J$ contains a path from $j_1$ to $j_k$;
	
	\item $j_k$ only has one out-neighbour in $J$, namely $j_1$;
	
	\item $I$ and $J$ are connected using undirected edges as follows: $i_1 j_1$, $i_a j_b$ for all $1 \le a \le k-1$ and $2 \le b \le k-1$, and $i_c j_k$ for all $2 \le c \le k-1$.
\end{itemize}
A graph $G_3$ is illustrated in Figure \ref{fig:non-LS}, where we have chosen the graph which included all possible arcs.

\begin{figure}[Ht]
\begin{center}
\begin{tikzpicture}
 \tikzstyle{every node}=[draw,shape=circle];

	\node (i1) at (0,2) {$i_1$};
	\node (i2) at (0,0) {$i_2$};
	\node (j1) at (3,4) {$j_1$};
	\node (j2) at (2,2) {$j_2$};
	\node (j3) at (3,0) {$j_3$};

	\draw[-latex] (i1) -- (i2);
	
	\draw[-latex] (j1) -- (j2);
	\draw[-latex] (j1) -- (j3);
	\draw[-latex] (j2) -- (j3);
	\draw[-latex] (j3) -- (j1);
	
	\draw[latex-latex] (i1) -- (j1);
	
	\draw[latex-latex] (i1) -- (j2);
	\draw[latex-latex] (i2) -- (j2);

	\draw[latex-latex] (i2) -- (j3);
\end{tikzpicture}
\end{center}
\caption{Example of a non linearly solvable digraph: $G_3$.} \label{fig:non-LS}
\end{figure}
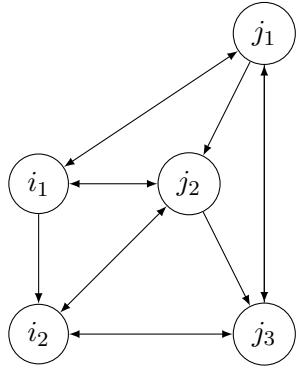

\begin{theorem} \label{th:non_solvable_digraph}
For any $k \ge 2$, $k(G_k) = k$ and $g_L(G_k,q) = k-1$ for all $q$.
\end{theorem}

\begin{proof}
We first verify that $I$ is a maximum acyclic set, i.e. that no set of $k$ vertices is acyclic. Let $S$ be a set of $k$ vertices. If $S = J$, $S$ is not acyclic. Now suppose $S$ contains a vertex $i \in I$. Firstly, suppose $i = i_1$, then if $S$ contains $j_b$ for $1 \le b \le k-1$, $S$ contains the cycle $i_1j_b$, otherwise the only case left is $S = I \cup \{j_k\}$, which again has a cycle $i_{k-1} j_k$. Secondly, suppose $i = i_a$ for $2 \le a \le k-1$, then if $S$ contains $j_b$ for $2 \le b \le k$, $S$ contains the cycle $i_aj_b$, otherwise the only case left is $S = I \cup \{j_1\}$, which has a cycle $i_{1} j_1$.


Now suppose that $g_L(G_k,q)=k$, that is, $G$ is linearly solvable for some $q$. Then it has a strictly linearly solvable subgraph $H$ such that $k(H)=k$. Then we force $(j_k,j_1) \in H$ because if $(j_k,j_1) \notin H$, then $I$ becomes a feedback vertex set of $H$ of size $k-1$. Now, by construction, $H$ has no path from $j_k$ to $j_1$ through $I$; thus $I$ is a maximum acyclic set of $H$ which is not weakly compatible, thus by Theorem \ref{th:SLS_weakly} $H$ is not strictly linearly solvable, a contradiction. Thus $g_L(G_k,q) \le k-1$. Conversely, $G$ contains a matching of size $k-1$, namely $\{i_aj_a : 1 \le a \le k-1\}$, thus $g_L(G_k,q) = k-1$ for all $q$.
\end{proof}

\subsection{Strictly linearly solvable graphs} \label{sec:SLS}

The reduction of coding functions also allows to construct strictly linearly solvable digraphs. Theorem~\ref{th:SLS_weakly} and its applications to Theorems \ref{th:triangle_free} and \ref{th:non_solvable_digraph} already illustrated how to use strictly linearly solvable digraphs as a means to study linearly solvable graphs. Nonetheless, we would like to motivate the study of strictly (linearly) solvable digraphs. Firstly, in the context of (Boolean) coding functions used as models of gene networks, an arc $(u,v)$ in the interaction graph illustrates the fact that the gene $u$ directly influences the gene $v$: such an influence may not be ignored. Secondly, studying strictly solvable digraphs indicates which arcs must be ignored in order to correctly transmit information by network coding. Indeed, suppose $G$ is solvable but not strictly solvable, then there exists an arc in $G$ which must not be used in any solution of $G$. Therefore, that arc is not only useless, but it is actually detrimental to network coding. An example is given in Figure \ref{fig:non-strictly}; the graph is clearly solvable, yet the thick arc makes it non-strictly solvable. Thirdly, by focusing on strictly linearly solvable digraphs, we show in Theorem \ref{th:SLS_strongly} that a large class of digraphs are linearly solvable.

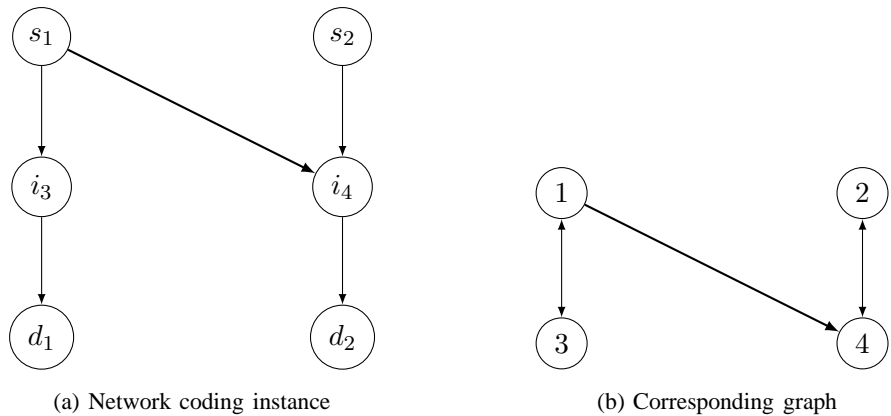
\begin{figure}
\centering
	\subfloat[Network coding instance]
	{\begin{tikzpicture}[auto]

	\node[draw,shape=circle] (s1) at (0,4) {$s_1$};
	\node[draw,shape=circle] (d1) at (0,0) {$d_1$};
	
	\node[draw,shape=circle] (s2) at (4,4) {$s_2$};
	\node[draw,shape=circle] (d2) at (4,0) {$d_2$};
	
	\node[draw,shape=circle] (i3) at (0,2) {$i_3$};
	\node[draw,shape=circle] (i4) at (4,2) {$i_4$};

	\draw[-latex] (s1) -- (i3);
	\draw[-latex] (i3) -- (d1);
	
	\draw[-latex, thick] (s1) -- (i4);
	
	\draw[-latex] (s2) -- (i4);
	\draw[-latex] (i4) -- (d2);
	
	\end{tikzpicture}} \hspace{2cm} \subfloat[Corresponding graph]
	{\begin{tikzpicture}	
	
	\node[draw,shape=circle] (1) at (0,4) {$1$};
	
	\node[draw,shape=circle] (2) at (4,4) {$2$};
	
	\node[draw,shape=circle] (3) at (0,2) {$3$};
	\node[draw,shape=circle] (4) at (4,2) {$4$};

	\draw[latex-latex] (1) -- (3);
	
	\draw[-latex, thick] (1) -- (4);
	
	\draw[latex-latex] (2) -- (4);
	\end{tikzpicture}
	}
\caption{A non-strictly solvable network coding instance.}
\label{fig:non-strictly}
\end{figure}

\begin{theorem} \label{th:SLS_strongly}
If $G$ has a strongly compatible maximum acyclic set and no loop, then $G$ is strictly linearly solvable.
\end{theorem}

\begin{proof}
The reader is reminded of the ${\mathbbm 1}\{\mathcal{P}\}$ notation used in the proof of Theorem \ref{th:h_loops}. Let $I$ be a strongly compatible maximum acyclic set, say $I = \{i_1,\dots,i_m\}$ in topological order. For all vertices $u,v\in V\setminus I$, the number of path from $u$ to $v$ through $I$ is denoted $N_I(u,v)$. Let $q$ be a prime number greater than $\max_{u,v\in V\setminus I} N_I(u,v)$, and $f \in L(G,q)$ as follows:
\begin{align*}
	f_i(x) &= \sum_{u \in \inn(v)} x_u, \quad \forall i \in I\\
	f_v(x) &= \sum_{i \in \inn(v) \cap I} \frac{1}{N_I(v,v)} x_i - \sum_{j \in \inn(v) \setminus I} \frac{N_I(j,v)}{N_I(v,v)} x_j \quad \forall v \notin I.
\end{align*}
Remark that $N_I(v,v) \ne 0$ since $V \setminus I$ is a minimal feedback vertex set and that $N_I(j,v) \ne 0$ since $I$ is strongly compatible. The inverse of $N_I(v,v)$ then exists since $q$ is a prime; since $q$ is larger than any $N_I(j,v)$, we have $N_I(j,v) \ne 0 \mod q$ either. Therefore, $G(f) = G$. 

We shall prove that $f^{-I}$ is the identity. For that purpose, we prove the following by induction on $0 \le b \le m$. Let $I_0 = \emptyset$ and $I_b = \{i_1,\dots,i_b\}$, then
\begin{align*}
	f_i^{-I_b}(x) &= \sum_{u \notin I_b} (N_{I_b}(u,i) + {\mathbbm 1}\{(u,i) \in G\}) x_u, \quad \forall i \in I \setminus I_b\\
	f_v^{-I_b}(x) &= \sum_{u \notin I_b} \frac{N_{I_b}(u,v)}{N_I(v,v)} x_u + \sum_{i \in \inn(v) \cap (I \setminus I_b)} \frac{1}{N_I(v,v)} x_i - \sum_{j \in \inn(v) \setminus I} \frac{N_I(j,v)}{N_I(v,v)} x_j \quad \forall v \notin I.	
\end{align*}

This clearly holds for $b=0$. We have $f^{-I_{b+1}} = f^{-I_b -i_{b+1}}$. Let $i \in I \setminus I_{b+1}$; by induction hypothesis, the $x_u$ term in $f_i^{-I_b}$ is
\begin{equation} \label{eq:Ib1}
	f_i^{-I_b}(x_u) = N_{I_b}(u,i) + {\mathbbm 1}\{(u,i) \in G\};
\end{equation}
the $x_{i_{b+1}}$ term in $f_i^{-I_b}$ is
$$
	f_i^{-I_b}(x_{i_{b+1}}) = N_{I_b}(i_{b+1},i) + {\mathbbm 1}\{(i_{b+1},i) \in G\} = {\mathbbm 1}\{(i_{b+1},i) \in G\};
$$
and the $x_u$ term in $f_{i_{b+1}}^{-I_b}$ is
$$
	f_{i_{b+1}}^{-I_b}(x_u) = N_{I_b}(u,i_{b+1}) + {\mathbbm 1}\{(u,i_{b+1}) \in G\}.
$$
By applying the reduction, we obtain that the $x_u$ term in $f_i^{-I_{b+1}}$ is
\begin{align*}
	f_i^{-I_{b+1}}(x_u) &= {\mathbbm 1}\{(u,i) \in G\} + N_{I_b}(u,i) + {\mathbbm 1}\{(i_{b+1},i) \in G\} \left( N_{I_b}(u,i_{b+1}) + {\mathbbm 1}\{(u,i_{b+1}) \in G\}\right)\\
	&= {\mathbbm 1}\{(u,i) \in G\} + N_{I_{b+1}}(u,i).
\end{align*}

Now let $v \notin I$: we have two cases to consider for $f_v^{-I_{b+1}}$. First, let $i \in I\setminus I_b $; by induction hypothesis, the $x_i$ term in $f_v^{-I_b}$ is
\begin{equation} \label{eq:Ib2}
	f_v^{-I_b}(x_i) = \frac{1}{N_I(v,v)} \left[ N_{I_b}(i,v) + {\mathbbm 1}\{(i,v) \in G\} \right] = \frac{{\mathbbm 1}\{(i,v) \in G\}}{N_I(v,v)},
\end{equation}
since there is no path from $i$ to $v$ through $I_b$; and the $x_i$ term in $f_{i_{b+1}}^{-I_b}$ is
$$
	f_{i_{b+1}}^{-I_b}(x_i) = N_{I_b}(i,i_{b+1}) + {\mathbbm 1}\{(i,i_{b+1}) \in G\} = 0,
$$
for similar reasons. By applying the reduction, we obtain that the $x_i$ term in $f_v^{-I_{b+1}}$ is
$$
	f_v^{-I_{b+1}}(x_i) = \frac{ {\mathbbm 1}\{(i,v) \in G\}}{N_I(v,v)}.
$$
Secondly, let $u \notin I$;  by induction hypothesis, the $x_u$ term in $f_v^{-I_b}$ is
$$
	f_v^{-I_b}(x_u) = \frac{1}{N_I(v,v)} \left[ N_{I_b}(u,v) - N_I(u,v){\mathbbm 1}\{(u,v) \in G\} \right];
$$
the $x_{i_{b+1}}$ term in $f_v^{-I_b}$ is
$$
	f_v^{-I_b}(x_{i_{b+1}}) = \frac{{\mathbbm 1}\{(i_{b+1},v) \in G\}}{N_I(v,v)} 
$$
(similarly to \eqref{eq:Ib2}) and the $x_u$ term in $f_{i_{b+1}}^{-I_b}$ is
$$
	f_{i_{b+1}}^{-I_b}(x_u) = N_{I_b}(u,i_{b+1}) + {\mathbbm 1}\{(u,i_{b+1}) \in G\}
$$
(similar to \eqref{eq:Ib1}). By applying the reduction, we obtain that the $x_u$ term in $f_v^{-I_{b+1}}$ is
\begin{align*}
	f_v^{-I_{b+1}}(x_u) &= \frac{1}{N_I(v,v)} \left[ N_{I_b}(u,v) + {\mathbbm 1}\{(i_{b+1},v) \in G\} \left( N_{I_b}(u,i_{b+1}) + {\mathbbm 1}\{(u,i_{b+1}) \in G\} \right) - N_I(u,v){\mathbbm 1}\{(u,v) \in G\} \right]\\
	&= \frac{N_{I_{b+1}}(u,v) - N_I(u,v) {\mathbbm 1}\{(u,v) \in G\}}{N_I(v,v)}.
\end{align*}

Having proved the claim, we can use it for $b = m$: this yields
$$
	f_v^{-I}(x) = \sum_{u \notin I} \frac{N_I(u,v) - {\mathbbm 1}\{(u,v) \in G\}N_I(u,v)}{N_I(v,v)} x_u. 
$$
The $x_v$ term in $f_v^{-I}(x)$ is then $1$ (since $G$ has no loop on $v$); if $u \in \inn(v)$, the term is $(N_I(u,v) - N_I(u,v))/N_I(v,v)=0$; if $u \notin \inn(v)$, we have $N_I(u,v) = 0$ since $I$ is strongly compatible and hence the term in $x_u$ also vanishes. Thus, $f^{-I}_v(x) = x_v$ .
\end{proof}

\begin{corollary} \label{cor:induced}
For any loopless digraph $D$, there exists a strictly linearly solvable graph $G$ such that $D$ is an induced subgraph of $G$.
\end{corollary}

\begin{proof}
We shall use a construction similar to that in the proof of Proposition \ref{prop:D_H}. Let $J$ be the vertex set of $D$, then let $G$ be the graph with $G[J] = D$ and such that, for any arc $(u,v)$ of $D$, $G$ contains $|J| + 1$ vertices $(u,v,1), \dots, (u,v,|J|+1)$ and the arcs $(u,(u,v,i))$ and $((u,v,i), v)$ for all $1 \le i \le |J| + 1$. Then the vertices outside of $J$ form a strongly compatible maximum acyclic set and $G$ is strictly linearly solvable.
\end{proof}

Corollary \ref{cor:induced} indicates that non-solvability is not a local property. One cannot isolate an induced subgraph of a graph and decide that this graph is not solvable.

Note that the converse of Theorem \ref{th:SLS_strongly} is not true: the complete bipartite graph $K_{2,2}$, illustrated in Figure \ref{fig:K22} is strictly linearly solvable but does not have any strongly compatible maximum independent set. The strict solution for $K_{2,2}$ is given, for any odd field, by
\begin{align*}
	f_1(x_3,x_4) &= \frac{x_3 + x_4}{2},\\
	f_2(x_3,x_4) &= \frac{x_3 - x_4}{2},\\
	f_3(x_1,x_2) &= x_1 + x_2,\\
	f_4(x_1,x_2) &= x_1 - x_2.	
\end{align*}
Notably, the graph $G$ on Figure \ref{fig:non-strictly} is a subgraph of $K_{2,2}$. Therefore, the thick arc, which is detrimental in $G$, becomes useful in $K_{2,2}$.

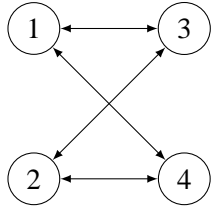
\begin{figure}[Ht]
\begin{center}
\begin{tikzpicture}
 \tikzstyle{every node}=[draw,shape=circle];

	\node (1) at (0,2) {1};
	\node (2) at (0,0) {2};
	\node (3) at (2,2) {3};
	\node (4) at (2,0) {4};

	\draw[latex-latex] (1) -- (3);
	\draw[latex-latex] (1) -- (4);
	\draw[latex-latex] (2) -- (3);
	\draw[latex-latex] (2) -- (4);
\end{tikzpicture}
\end{center}
\caption{$K_{2,2}$: a strictly linear solvable graph which is not edge-full.} \label{fig:K22}
\end{figure}

We generalise this observation to all balanced complete bipartite graphs.

\begin{proposition} \label{prop:K_kk}
The complete bipartite graph $K_{k,k}$ with $k \ge 1$ is strictly linearly solvable over all sufficiently large finite fields.
\end{proposition}

\begin{proof}
The case $k=1$ being clear, we assume $k \ge 2$ henceforth. We first prove that for all prime power $q \ge 3k^2$, there exists a $k \times k$ nonsingular matrix $M \in \GL(k,q)$ such that $M$ and $M^{-1}$ have no zero entry. Denote the set of $k \times k$ matrices over $\GF(q)$ with no zero entry as $Z(k,q)$; we then have
$$
	|Z(k,q)| = (q-1)^{k^2} \ge q^{k^2} \left( 1- \frac{k^2}{q} \right) \ge \frac{2}{3}q^{k^2},
$$
while the number of nonsingular matrices is famously lower bounded by 
$$
	|\GL(k,q)| \ge q^{k^2} \prod_{j = 1}^\infty (1 - q^{-j}) = q^{k^2} \sum_{l \in \mathbb{Z}} (-1)^l q^{-l(3l-1)/2} \ge q^{k^2} (1 - q^{-1} - q^{-2}) \ge \frac{9}{10} q^{k^2},
$$
using Euler's pentagonal number theorem. We obtain
$$
	|\GL(k,q) \cap Z(k,q) | \ge q^{k^2} \left( \frac{9}{10} + \frac{2}{3} - 1 \right) > \frac{1}{2}q^{k^2} > \frac{1}{2}|\GL(k,q)|.
$$
Hence $|\GL(k,q) \cap Z(k,q)| > |\GL(k,q) \setminus Z(k,q)|$. Thus, inversion cannot be an injection from $\GL(k,q) \cap Z(k,q)$ to $\GL(k,q) \setminus Z(k,q)$ and such a matrix $M$ exists.

Now let $q \ge 3k^2$ and $M$ such that $M, M^{-1} \in Z(k,q)$. Let the vertex set of $K_{k,k}$ be $L \cup R$ and consider the following linear coding function on $K_{k,k}$:
$$
	f_R(x_L) = x_L M, \quad f_L(x_R) = x_R M^{-1}.
$$
Then clearly every vector of the form $(x_L, x_R = x_L M)$ is fixed by $f$.
\end{proof}

We make a note on edge-full undirected graphs. An {\bf intersection model} for an undirected graph $G$ is an ordered pair $(S,X)$, where  $S$ is a set and $X = (X_1, \dots, X_n)$ is a collection of $n$ subsets of $S$ such that for all vertices $u, v$ of $G$, $uv$ is an edge if and only if $X_u \cap X_v \ne \emptyset$. The size of the intersection model is simply the size of $S$; the minimum size of an intersection model for $G$ is denoted as $\epsilon(G)$. Then $\epsilon(G) \ge \alpha(G)- i(G)$, where $i(G)$ is the number of isolated vertices of $G$. Indeed, any non-isolated vertex in a maximum independent set needs at least a singleton in the model; all these are disjoint, hence any intersection model must have at least that many elements.

\begin{proposition} \label{prop:edge-full}
Let $G$ be an undirected graph. Then the following are equivalent.
\begin{enumerate}
	\item \label{it:un_compatible} An independent set of $G$ is strongly compatible.

	\item \label{it:un_some} A maximum independent set of $G$ is strongly compatible.
	
	\item \label{it:un_all} All maximum independent sets of $G$ are strongly compatible. 

	\item \label{it:un_cover_edges} $G$ is edge-full.
	
	\item \label{it:un_model} $\epsilon(G) = \alpha(G) - i(G)$.
\end{enumerate}
\end{proposition}

\begin{proof}
Clearly, Property \ref{it:un_all} implies \ref{it:un_some}, which in turn implies \ref{it:un_compatible}. Moreover, if $I = \{i_1,\dots,i_m\}$ is a strongly compatible independent set (non necessarily maximum), then the neighbourhood of each $i_l$ is a clique. We claim that these $m$ cliques cover all edges in $G$. Indeed, there are no edges in $I$; any edge with one vertex in $I$ is clearly covered by these cliques; finally, for any edge $uv$ outside of $I$, then there is $i \in I$ such that $uv$ is in the clique corresponding to $i$. Conversely, it clearly takes at least $\alpha(G)$ cliques to cover all the vertices of $G$, and hence at least $\alpha(G)$ cliques to cover all edges of $G$. This shows that any strongly compatible independent set is maximum and that Property \ref{it:un_compatible} implies \ref{it:un_cover_edges}.

We now show that Property \ref{it:un_cover_edges} implies \ref{it:un_all}. Suppose all edges of $G$ are covered by $\alpha(G)$ cliques $c_1,\dots,c_\alpha$, then any maximum independent set $I$ contains one vertex $i_1, \dots, i_\alpha$ per clique; clearly, each $i_l$ belongs to exactly one clique $c_l$. Suppose $u$ and $v$ are vertices outside of $I$. If $uv$ is an edge, then it belongs to some clique $c_\beta$ and hence $u i_\beta, i_\beta v$ are edges in $G$. Conversely, if $uv$ is not an edge, then $u$ and $v$ cannot belong to a common clique and hence there is no vertex $i \in I$ such that $ui, iv$ are edges. Thus, $I$ is strongly compatible.

Clearly, Property \ref{it:un_some} implies \ref{it:un_model}: if $I$ is a strongly compatible maximum independent set, then let $S = I\setminus U$, with $U$ the set of isolated vertices of $G$, and $X_v = (v \cup \inn(v)) \cap S$. Conversely, if $G$ has a model $(S = \{s_1,\dots, s_\epsilon\}, X)$ of size $\epsilon = \alpha(G) - i(G)$, then if $I$ is a maximum independent set, we must have an enumeration $\{i_1,\dots,i_\epsilon\}$ of $I \setminus U$ such that $X_{i_1} = s_1, \dots, X_{i_\alpha} = s_\epsilon$. Thus, for any $u,v \notin I$, $uv$ is an edge if and only if $s_\beta \in X_u \cap X_v$ for some $\beta$, which is equivalent to $u i_\beta$ and $i_\beta v$ being edges, and $I$ is strongly compatible.
\end{proof}

We give an example of a digraph which is not edge-full and yet is strictly linearly solvable in Figure \ref{fig:example}. The set $\{1,2\}$ is a strongly compatible maximum acyclic set, hence by Theorem \ref{th:SLS_strongly} the graph is strictly linearly solvable. Since the graph is not undirected, it is not edge-full; moreover, we remark that $\{1,5\}$ is a maximum acyclic set which is not strongly compatible.

\begin{figure}[Ht]
\begin{center}
\begin{tikzpicture}
 \tikzstyle{every node}=[draw,shape=circle];

	\node (1) at (0,2) {1};
	\node (2) at (0,0) {2};
	\node (3) at (3,4) {3};
	\node (4) at (2,2) {4};
	\node (5) at (3,0) {5};

	\draw[-latex] (1) -- (2);
	\draw[latex-latex] (1) -- (3);
	\draw[latex-latex] (1) -- (4);
	
	\draw[-latex] (5) -- (1);
	\draw[-latex] (2) -- (5);

	\draw[latex-latex] (3) -- (4);
	\draw[latex-latex] (3) -- (5);
	\draw[latex-latex] (4) -- (5);
\end{tikzpicture}
\end{center}
\caption{Example of a non edge-full digraph which is strictly linearly solvable.} \label{fig:example}
\end{figure}
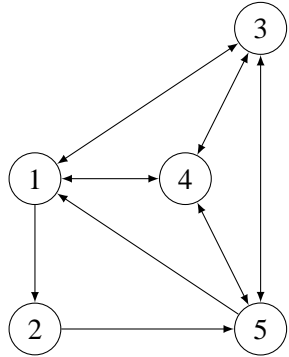

\section{Acknowledgment}

The authors would like to thank George Mertzios for interesting discussions leading to Proposition \ref{prop:edge-full}.

\bibliographystyle{IEEEtran}
\bibliography{g}

\end{document}